%% file: icml22.tex
\documentclass{article}

\usepackage{microtype}
\usepackage{graphicx}
\usepackage{subfigure}
\usepackage{booktabs} 

\usepackage{hyperref}

\usepackage[accepted]{icml2022}

\usepackage{amsmath}
\usepackage{amssymb}
\usepackage{mathtools}
\usepackage{amsthm}

\providecommand{\tightlist}{%
  \setlength{\itemsep}{0pt}\setlength{\parskip}{0pt}}

\usepackage{enumitem}

\usepackage[capitalize,noabbrev]{cleveref}

\theoremstyle{plain}
\newtheorem{theorem}{Theorem}[section]
\newtheorem{proposition}[theorem]{Proposition}
\newtheorem{lemma}[theorem]{Lemma}

\theoremstyle{definition}
\newtheorem{definition}[theorem]{Definition}

\theoremstyle{remark}

\usepackage[disable, textsize=tiny]{todonotes}

\newcommand{\LB}[1]{\todo{LB: #1}}
\input{include/notations.tex}
\newcommand{\mparagraph}[1]{\paragraph{#1}}

\newcommand{\optnl}{}
\newcommand{\optmult}[1]{\begin{equation*}#1\end{equation*}}
\newcommand{\optspace}[1]{}
\newcommand{\optesp}{}
\newcommand{\noicml}{\iffalse}
\newcommand{\foricml}{\if11}
\newcommand{\isicml}{1}

\icmltitlerunning{Fictitious Play and Best-Response Dynamics in Identical Interest and Zero Sum Stochastic Games}

\begin{document}

\twocolumn[
\icmltitle{Fictitious Play and Best-Response Dynamics \\ in Identical Interest and Zero Sum Stochastic Games}



\icmlsetsymbol{equal}{*}

\begin{icmlauthorlist}
\icmlauthor{Lucas Baudin}{yyy}
\icmlauthor{Rida Laraki}{yyy}
\end{icmlauthorlist}

\icmlaffiliation{yyy}{Université Paris-Dauphine - PSL, France}

\icmlcorrespondingauthor{Lucas Baudin}{lucas.baudin@dauphine.eu}

\icmlkeywords{Machine Learning, Stochastic Games, Fictitious Play, Best-Response Dynamics}

\vskip 0.3in
]



\printAffiliationsAndNotice{\icmlEqualContribution} 

\begin{abstract}
\input{include/abstract_fp_br}
\end{abstract}

\section{Introduction}
\label{sec:introduction}

\input{include/intro.tex}

\section{Background}
\label{sec:background}

\input{include/preliminaries.tex}

\section{Related Work}
\label{sec:related_work}

\input{include/related_work.tex}

\section{Discrete Time : Fictitious Play Procedures}
\label{sec:fictitious_play}

\input{include/fp_identical_interest.tex}

\section{Continuous Time: Best-Response Dynamics}
\label{sec:brd}

\input{include/best_response_dynamics.tex}

\section{Linking Continuous and Discrete Systems}
\label{sec:link}

\input{include/link_continuous_discrete.tex}

\section{Conclusion}
\label{sec:conclusion}
\input{include/conclusion.tex}

\bibliography{these}
\bibliographystyle{icml2022}

\newpage
\appendix
\onecolumn

\section{Convergence of Discrete-Time Fictitious Play in Identical Interest Stochastic Games}\label{app:discreteproof}

\input{include/appendix_direct_discrete_proof.tex}

\section{Convergence of Best-Response Dynamics in Identical Interest Stochastic Games}
\label{app:continuous_proof}
\input{include/appendix_continuous.tex}

\section{Convergence of Best-Response Dynamics in Zero-Sum Games}
\label{app:zerosum_continuous}

\input{include/appendix_zerosum_continuous.tex}

\section{Using Continuous Time Results for Discrete Time Fictitious Play with Stochastic Approximations}
\label{app:stochasticapprox}

\input{appendix_stochastic.tex}



\section{Technical Lemmas}
\label{app:technical}
\input{include/appendix_technical_lemmas.tex}

\end{document}

%% file: include/notations.tex
\newcommand{\actions}[0]{A^i}

\newcommand{\transSS}[0]{P_{ss'}}
\newcommand{\ris}[0]{r^i_s}

\newcommand{\dxi}[2]{x^i_{#1, #2}}
\newcommand{\dui}[2]{u^i_{#1, #2}}
\newcommand{\dxin}[1]{x^i_{n, #1}}
\newcommand{\dxn}[2]{x^{#1}_{n, #2}}

\newcommand{\dxinp}[1]{x^i_{n+1, #1}}
\newcommand{\duin}[1]{u^i_{n,#1}}
\newcommand{\duinp}[1]{u^i_{n+1,#1}}

\newcommand{\stfis}[2]{f^i_{s, #1}(#2)}
\newcommand{\duins}[0]{u^i_{n, s}}
\newcommand{\dxins}[0]{x^i_{n, s}}

\newcommand{\countsn}[0]{s^{\sharp}_n}
\newcommand{\countsnp}[0]{s^{\sharp}_{n-1}}
\newcommand{\payt}[2]{\Gamma_{#1}(#2)}
\newcommand{\pay}[1]{\Gamma_{#1}}
\newcommand{\wstate}[1]{s_{-}(#1)}
\newcommand{\cut}[0]{u(t)}
\newcommand{\stfs}[2]{f_{s, #1}(#2)}
\newcommand{\cxs}[2]{x^{#1}_s(#2)}
\newcommand{\limt}[0]{\xrightarrow[t\rightarrow \infty] {}}
\newcommand{\ratest}[2]{\beta_{#1}(#2)}

\newcommand{\states}[0]{S}

\newcommand{\crateu}[1]{\alpha\left(#1\right)}
\newcommand{\ratelb}[0]{\beta_{-}}
\newcommand{\indic}[2]{1_{#1=#2}}
\newcommand{\startt}[0]{0}
\newcommand{\nlog}[0]{\overline{\log}}
\newcommand{\pos}[1]{\left( #1\right)_{+}}
\newcommand{\diff}[1]{\frac{d #1}{dt}}
\newcommand{\paysup}[0]{\Psi}

\newcommand{\ql}[0]{$Q$-learning}
\newcommand{\argmin}[1]{\underset{#1}{\operatorname{arg\,min}}}

\newcommand{\argmax}[1]{\underset{#1}{\operatorname{arg\,max}}\,}
\newcommand{\argmaxi}[1]{\operatorname{arg\,max}_{#1}}
\newcommand{\eref}[1]{(\ref{#1})}
\newcommand{\sax}{y}

\usepackage{xparse}
\usepackage{ifthen}
\NewDocumentCommand{\mcus}{ O{s} O{t} }{u_{#1}(#2)}
\NewDocumentCommand{\mdus}{ O{s} O{n} O{} }{u^{#3}_{#2\ifthenelse{\equal{#1}{}}{}{, #1}}}
\NewDocumentCommand{\mpayt}{ O{s} O{t} }{\Gamma_{#1}(#2)}
\NewDocumentCommand{\mdual}{ O{s} O{t} }{w_{#1}(#2)}
\NewDocumentCommand{\mdfsu}{ O{x_{#3,#2}} O{s} O{n} O{u_{#3}} }{f_{#2, #4}(#1)}
\NewDocumentCommand{\mfsu}{ O{x_{#2}(#3)} O{s} O{t} O{u} }{f_{#2, #4}(#1)}
\NewDocumentCommand{\mdxsn}{O{} O{s} O{n}}{x^{#1}_{#3, #2}}
\NewDocumentCommand{\bestr}{ O{i} O{s} O{u} O{x^{-#1}_{#2}}}{\text{br}^{#1}_{#2, #3}(#4)}

\NewDocumentCommand{\cxst}{O{} O{s} O{t}}{x^{#1}_{#2}\ifthenelse{\equal{#3}{}}{}{(#3)}}
\NewDocumentCommand{\optgap}{O{i} O{s} O{t}}{\Delta^{#1}_{#2}(#3)}
\NewDocumentCommand{\probatr}{O{} O{s} O{\cxst}}{P_{#2,#1}(#3)}
\NewDocumentCommand{\diffu}{O{s} O{t} O{\crateu #2}}{\frac{\mpayt[#1][#2] - \mcus[#1][#2]}{\crateu{#2}}}

\NewDocumentCommand{\fgap}{O{s} O{n}}{\Gamma_{#2, #1}}

\NewDocumentCommand{\durate}{O{n}}{\alpha_{#1}}
\NewDocumentCommand{\sdurate}{O{n}}{\sigma_{#1}}

%% file: include/abstract_fp_br.tex
This paper proposes an extension of a popular decentralized discrete-time learning procedure when repeating a static game called fictitious play (FP) \cite{brown1951iterative, robinsonIterativeMethodSolving1951} to a dynamic model called discounted stochastic game \cite{shapleyStochasticGames1953}. Our family of discrete-time FP procedures is proven to converge to the set of stationary Nash equilibria in identical interest discounted stochastic games. 
This extends similar convergence results for static games \cite{mondererFictitiousPlayProperty1996}.
We then analyze the continuous-time counterpart of our FP procedures, which include as a particular case the best-response dynamic introduced and studied by~\citet{leslieBestresponseDynamicsZerosum2020} in the context of zero-sum stochastic games. We prove the converge of this dynamics to stationary Nash equilibria in identical-interest and zero-sum discounted stochastic games. Thanks to stochastic approximations, we can infer from the continuous-time convergence some discrete time results such as the convergence to stationary equilibria in zero sum and team stochastic games \cite{hollerLearningDynamicsReinforcement2020}.



%% file: include/intro.tex
\textit{Learning Nash equilibria of a static game} $G$ after playing repeatedly $G$ is a subject that has been widely studied  almost since the beginning of game theory \citep{fudenbergTheoryLearningGames1998, hartSimpleAdaptiveStrategies2013, youngStrategicLearningIts2004, cesa-bianchiPredictionLearningGames2006,brown1951iterative,robinsonIterativeMethodSolving1951}. 
In contrast, \textit{learning Nash equilibria of a dynamic model such as a stochastic game} has comparatively been much less developed.
Some noticeable exceptions are the two recent papers ~\citep{leslieBestresponseDynamicsZerosum2020, sayinFictitiousPlayZerosum2020} which develop quite similar systems that converge to stationary equilibria in zero-sum discounted stochastic game. Our paper extends the continuous-time dynamics of ~\citep{leslieBestresponseDynamicsZerosum2020} to a large parametric class and to discrete-time. We show the convergence of the two systems to stationary equilibria in identical interest and also in zero-sum stochastic games.

Discounted stochastic games, introduced by \citet{shapleyStochasticGames1953},
model strategic interactions between players with a state variable.
Thus, compared to non-stochastic games, actions that players take impact
their current payoff but also a state variable that may influence their
future payoff. Therefore, this class of games offers a rich framework
\citep{neymanStochasticGamesApplications2003} that is especially well
suited for economic applications (see the survey by
\citet{amirStochasticGamesEconomics2003} and references therein), or
engineering applications. In the latter, this belongs to the more general framework of multi-agent reinforcement learning (see \citet{busoniuComprehensiveSurveyMultiagent2008} for a survey).

\citet{finkEquilibriumStochasticPerson1964} proved that when there are finitely many players and actions (our framework), any discounted stochastic game admits a stationary (mixed) Nash equilibrium (e.g. a decentralized randomized policy that depends only on the state variable). The proof is implied by the fact that a stationary equilibrium is the fixed point of an operator where, given other players stationary strategies, each player strategy is an optimal stationary policy in a Markov Decision Process. A stochastic game admits typically many other Nash equilibria. To see why, if there is only one state, it is a repeated game and a stationary equilibrium consists on playing independently the same mixed strategy. Therefore, the set of stationary Nash equilibrium payoffs of the discounted repeated game coincides with the Nash equilibria of its stage game (which is known to be a semi-algebraic set with several connected components \cite{larakiMathematicalFoundationsGame2019}). On the other hand, the famous folk theorem of repeated game \cite{aumannLongTermCompetitionGameTheoretic1994, 10.2307/1911307,larakiMathematicalFoundationsGame2019} shows that when the discount factor is large enough, any feasible and individually rational payoff of the stage game is a Nash equilibrium payoff of the repeated game.

Stationary Nash equilibria are the simplest of all equilibria and it is desirable to construct some natural learning procedures that are provable to converge to them. This is a challenging problem even in identical interest or zero-sum stochastic games (the subject of our paper).  The reason is that the payoff function of a player in a discounted stochastic game is not linear, nor concave or quasi-concave when the players are restricted to their stationary strategies and thus, no gradient method is guaranteed to converge, even to a local Nash equilibrium \cite{daskalakisComplexityConstrainedMinmax2021}. 
Only recently, two extensions of the oldest of learning procedures --fictitious play (FP)-- have been proposed and proven to converge to stationary equilibria but only in zero-sum stochastic games \citep{leslieBestresponseDynamicsZerosum2020, sayinFictitiousPlayZerosum2020}. Our article extends one of the two papers in both discrete and continuous time and shows that our systems converge in both regimes to the set of stationary Nash equilibria in identical-interest and zero-sum discounted stochastic games. 

Our family of procedures combines classical fictitious play (a kind of myopic best reply in the local game given some prior constructed from empirical observations) with an elaborate rule in the spirit of $Q$-learning to update the expected continuation payoff, as in \citep{leslieBestresponseDynamicsZerosum2020, sayinFictitiousPlayZerosum2020}. $Q$-learning is a quite famous model-free rule which allows to learn the stationary policy in dynamic programming without even knowing the transition probabilities or the state space). The updating rule we study needs the knowledge of the model, however, as observed in the conclusion, they can be easily modified to become model-free. 



The combination of FP and $Q$-learning is not surprising. A major trend in recent years is the advent of efficient reinforcement
learning algorithms \citep{suttonReinforcementLearningIntroduction2018}.
 \(Q\)-learning
\citep{watkinsTechnicalNoteQLearning1992} is one of the most successful
model-free algorithms \citep{kondaActorCriticTypeLearning1999} with numerous extensions~\citep{NIPS2010_091d584f, NEURIPS2020_0d2b2061}. 
On the game theory side, fictitious play \cite{brown1951iterative, robinsonIterativeMethodSolving1951} is one of the most
studied procedures of learning in games.

The idea of combining concepts from \(Q\)-learning and fictitious play emerged only
recently with the work of
\citet{leslieBestresponseDynamicsZerosum2020,sayinFictitiousPlayZerosum2020} in the context of zero-sum stochastic games. In these
papers, a mechanism inspired by \(Q\)-learning is used to learn rewards
in future states and fictitious play (or in continuous time, the related
best response dynamics) is employed to choose the action in the current
state taking into account the interaction with other players. Typically, players learn at a
fast rate the actions of the other players in every state but compute
the future rewards at a comparatively slower pace. Our work 
extends \citet{leslieBestresponseDynamicsZerosum2020}, by introducing other procedures with several time-scales, and prove their convergence to the set of Nash equilibria in
identical interest and in zero-sum stochastic games, resulting in a decentralized algorithm for fully cooperative multi-agent reinforcement learning \cite{busoniuComprehensiveSurveyMultiagent2008}. In contrast with \citet{leslieBestresponseDynamicsZerosum2020, sayinDecentralizedQLearningZerosum2021}, our algorithm can use the same timescale to learn actions and future rewards, which is of practical interest in implementations. Importantly, we prove convergence to global Nash equilibria and not to local Nash equilibria.

\hypertarget{contributions.}{%
\paragraph{Contributions.}\label{contributions.}}
Our contributions are as follows:
\begin{itemize}
\item
We define procedures to play stochastic games in discrete time combining ideas from fictitious play and $Q$-learning. We prove their convergence to the set of stationary Nash equilibria in identical interest stochastic games. The proof is given directly in discrete time.
\item We define the continuous-time counterpart of our procedures. It results in a generalization of the  best-response dynamics of \citet{leslieBestresponseDynamicsZerosum2020} where the relative timescales of the estimation of continuation and of other players strategies is much less restricted. We prove the convergence of our continuous-time dynamics to the set of stationary equilibria in identical interest and in zero-sum stochastic games.
\item The convergence in continuous time allows us to prove some new convergence results in discrete time for zero-sum and team stochastic games \cite{hollerLearningDynamicsReinforcement2020} (a class which includes identical interest stochastic games).
\end{itemize}

\paragraph{Outline}Section~\ref{sec:background} gives initial definitions and assumptions. The next section describes related works. Then, Section~\ref{sec:fictitious_play} introduces two families of fictitious play procedures in discrete time  whose convergence is shown in identical interest stochastic games. The continuous time best response dynamics are defined in Section~\ref{sec:brd} together with their convergence in identical interest and in zero-sum stochastic games. The last section uses the continuous and discrete time results to infer convergence of the discrete time procedures in zero-sum and in team stochastic games \cite{hollerLearningDynamicsReinforcement2020}.

%% file: include/preliminaries.tex
\hypertarget{stochastic-games}{%
\paragraph{Stochastic games}\label{stochastic-games}}

We study dynamically interactive multi-agent systems based on the model of discounted stochastic games. Two
or more players can take actions over an infinite horizon. Players' actions
affect both their current stage payoffs but also the transition probability of the future
state, which is the second determining factor of the total discounted average payoff.
Therefore, compared to standard repeated games where there is no evolving state variable, stochastic games add a
layer of complexity: a player who wants to optimize its payoff should
strike a balance between the instantaneous payoff optimization and an
advantageous orientation of the state. As in \cite{leslieBestresponseDynamicsZerosum2020,sayinFictitiousPlayZerosum2020}, we will focus on
finite games where the state space, the action sets and the player
set are finite.  

\begin{definition}\label{def:sg} \textbf{Stochastic games} are tuples
\(G=(S, I, (A^i)_{i \in I}, (r^i_s)_{i\in I, s \in S}, (P_s)_{s\in S})\)
where \(S\) is the state space (a finite set), \(I\) is the finite set of
players, \(A^i\) is the finite action set of player \(i\),
\(A := \Pi_{i\in I} A^i\) is the set of action profiles,
\(r^i_s : A \rightarrow \mathbb{R}\) is the stage reward of player
\(i\),
and \(P_s: A \rightarrow \Delta(S)\) is the transition probability
map (where $\Delta(S)$ is the set of probability distributions on $S$).\end{definition}

\newcommand{\aname}{a}

A stochastic game is played in discrete time as follows: it starts in an
initial state \(s_0 \in S\) and at every time step  \(n \in \mathbb{N}\),
the system state is \(s_n\). Knowing the current state \(s_n\) and the past history of states and actions \((s_0,a_0,...,s_{n-1},a_{n-1})\), every player \(i \in I\) chooses, independently from the other players, an action
\(a^i_n \in A^i\) (potentially at random) 
and receives a stage reward of \(r^i_{s_n}(a_n)\). The
new state \(s_{n+1}\) is the realization of a random variable whose
distribution is \(P_{s_n}(a_n)\). The total payoff of such a sequence of
play for player \(i\) is \((1-\delta)\sum_{k=0}^{\infty} \delta^k r^i_{s_k}(a_k)\) where \(\delta \in (0, 1)\) is the discount factor (the $1-\delta$ factor is the usual normalization).

A behavioral strategy $\sigma^i$ for player $i$ is a mapping associating with 
each stage $n\in \mathbb{N}$, history $h_n\in (S\times A)^n$ and current state $s$, a mixed action $x^i_n=\sigma^i(n,h_n,s)$ in $\Delta(A^i)$. The behavioral strategy is pure if its image is always in $A^i$. By Kolmogorov's extension theorem, each behavioral strategy profile induces a unique probability distribution on the set of infinite histories, from which one can compute an expected discounted
payoff for every player \cite{neymanStochasticGamesApplications2003}. A stationary strategy of player \(i\) is the simplest of behavioral strategies. It depends only on the current state $s$ but not on the period $n$ nor on the past history $h_n$. As such, a stationary strategy can be identified with an element of \(\Delta(A^i)^S\) (a mixed action per state interpreted as: whenever the state is $s$, $i$ plays randomly according to distribution $x^i_s$). The set of stationary strategy profiles is $\Pi_{i\in I} \Delta(A^i)^S$. For \(y^i \in (\Delta(A^i))^S\) and \(x \in \Pi_{i\in I} \Delta(A^i)^S\), we denote by \((y^i, x^{-i})\) the stationary strategy profile where $i$ changes its strategy from $x^i$ to $y^i$. 

Extending a result of \cite{shapleyStochasticGames1953} for zero-sum games, \cite{finkEquilibriumStochasticPerson1964} proved the existence of stationary Nash equilibria in every finite discounted stochastic game, as a fixed point of the Shapley operator. Fink's characterization implies that a stationary strategy profile forms a Nash equilibrium iff there is no pure stationary profitable deviation. 




\begin{proposition}[Stationary equilibrium characterization]
A stationary profile \(x \in \Pi_{i\in I} \Delta(A^i)^S\) is a Nash equilibrium if and only if
no pure stationary deviation is profitable: for every player \(i\), its
expected total payoff with \(x\) is greater or equal than its expected
total payoff of strategy \((b, x^{-i})\) for any
\(b \in (A^i)^{S}\).\end{proposition}

Fink's result implies that (1) checking that a stationary profile is a Nash equilibrium is equivalent to solving finitely many polynomial inequalities, implying that the set of stationary strategy profiles is a semi-algebraic set \cite{neymanStochasticGamesApplications2003}; (2) a stationary profile is a Nash equilibrium of our stochastic games if and only if it is a Nash equilibrium of a
restricted stochastic game where each player $i$ is restricted to play a stationary strategy. While this restricted game is smooth (strategy spaces are convex/compact and the payoff functions analytic) the payoff functions are not linear, nor concave, nor quasi-concave with respect to a player own-strategy. This makes the computation of an equilibrium very hard: no learning gradient-based method is guaranteed to convergence, even to a local Nash equilibrium, and even in a zero-sum game \cite{daskalakisComplexityConstrainedMinmax2021}. But this is not any non-concave game: its dynamic nature is structured enough to allow learning to occur. Using the dynamic programming principle that stationary Nash equilibria satisfy allowed  \cite{leslieBestresponseDynamicsZerosum2020, sayinFictitiousPlayZerosum2020} to introduced some learning procedures that converge to stationary equilibria in zero-sum  stochastic games. We enlarge the set of procedures in \cite{leslieBestresponseDynamicsZerosum2020} and show convergence to stationary equilibria in identical interest and also zero-sum stochastic games. Formally, a stochastic game is of \emph{identical
interest} if all players have the same stage reward
function, i.e., for every state \(s \in S\), there is a function
\(r_s\) such that for every player \(i\), \(\ris = r_s\) and it is zero-sum when there are only two players 1 and 2 and $r^1_s+r^2_s=0$. 

When there is only one state (a repeated game) \citep{mondererFictitiousPlayProperty1996} proved that fictitious play converges to Nash equilibria of the stage game (hence to stationary equilibria of the repeated game). Their proof uses extensively the multi-linearity of the (common) payoff function. Since we don't have this multi-linearity nor do we have multi-concavity or multi-quasi-concavity, we must use the structure of the problem, namely that stationary equilibria satisfy a dynamic programming principle (each player is optimizing at every state $1-\delta$ its current payoff plus $\delta$ the expectation of the continuation payoff). As such, our procedures will use fictitious play at the local static game in which the continuation payoff vector is fixed and update the continuation payoff using a $Q$-learning like rule. But, since learning can only occur if all states are visited infinitely often, the following assumption is needed.  




\begin{definition}[Ergodicity]
A stochastic game is ergodic if there is a finite time $T$ such that for every $s$ and $s'$ there is a positive probability that the system starting from $s$ is in $s'$ after $T$ steps for any actions  taken.
\end{definition}




Some final notations. We denote by \(\transSS (\aname)\) the probability to go to state \(s'\) starting from \(s\) with action
profile \(\aname \in A\). As players will be randomizing, let the functions
\(P_{s s'}\) and \(\ris\) be multi-linearly extended to mixed action
profiles (i.e., $\Pi_{i\in I} \Delta(A^i)^S$) and are therefore \(I\)-linear. Finally, a stationary equilibrium payoff is an \(S\times I\)-vector that corresponds to a stationary Nash equilibrium.

%% file: include/related_work.tex
\paragraph{Fictitious Play}

Fictitious play (FP) is a procedure that was introduced in discrete time to play the same stage 
game repeatedly. It asks every player
to play a best response to a prior (a mixed strategy profile) which is equal to the empirically past actions of the opponents. It was initially proposed by \citet{brown1951iterative} and \citet{robinsonIterativeMethodSolving1951} who proved that when the stage game is a zero-sum game and both players use FP, the empirical distribution of actions converges to the set of Nash equilibria of the stage game. A similar result has been obtained when the stage game is a potential game
\citep{mondererFictitiousPlayProperty1996}, a \(2\times n\) game
\citep{bergerFictitiousPlayGames2005}, or a mean-field game \cite{perrinFictitiousPlayMean2020}.
Similar convergence results have been obtained for some variants
such as smooth FP, or vanishingly smooth FP, o stochastic FP
\citep{benaimConsistencyVanishinglySmooth2013,fudenbergConsistencyCautiousFictitious1995, fudenbergTheoryLearningGames1998,hofbauerGlobalConvergenceStochastic2002} 
But FP as well as all no-regret algorithms fail to converge to Nash equilibra in all finite games \cite{ shapleyTopicsTwoPersonGames1964,10.2307/3132156,hofbauer_sigmund_1998,DEMICHELIS2000192}.


FP is much less studied when the stage games vary (stochastic games) and there are many possible extensions. \citet{vriezeFictitiousPlayApplied1982} used a FP like algorithm to compute the value of a stochastic game but it is not a learning procedure (as it assumes all the states are observed and updated at every stage and so it does not define a behavioral strategy). \citet{perkinsAdvancedStochasticApproximation2013} defined and studied another FP procedure in zero-sum and identical interest stochastic games relying crucially on two-timescale updates and with some restrictions on the discount factor (for zero-sum games) and the structure of equilibria (for identical interest stochastic games). 
More recently, \citet{sayinFictitiousPlayZerosum2020} introduced another discrete-time variant of FP and deduced its convergence in zero-sum stochastic games from the convergence of an associated continuous-time best response dynamics explained below.  

\paragraph{Best-response dynamics} In continuous time, best-response dynamics \cite{harrisRateConvergenceContinuousTime1998a,matsuiBestResponseDynamics1992} is based on the same principle as fictitious play: each player adjusts its mixed action towards the best-response to the time-average mixed action of other players. For
repeated --non-stochastic-- games, it is the continuous-time
counterpart of the discrete-time fictitious play as the stochastic approximations framework \cite{benaimStochasticApproximationsDifferential2005} allows to deduce from the convergence of the continuous time dynamics the convergence of the (discrete-time) FP. Recently \citet{leslieBestresponseDynamicsZerosum2020} introduced an extension of the best-response dynamics to stochastis games with the continuation payoff updated at a slower pace in the spirit of $Q$-Learning (described below) and proved the convergence of this system to the set of stationary Nash equilibria of the underlying discounted stochastic game. However, \citet{leslieBestresponseDynamicsZerosum2020} did not prove the convergence to stationary equilibria of the discrete-time counterpart of their dynamics: we will do it in this article for their dynamics and others in zero-sum and also in identical interest stochastic games. \citet{sayinFictitiousPlayZerosum2020} defined an alternative dynamics close to the one of \citet{leslieBestresponseDynamicsZerosum2020} and using a two time-scale stochastic approximation theory, they show convergence of their discrete time FP to stationary equilibria of the zero-sum discounted stochastic game. In contrast, we give a direct convergence proof of our discrete-time system in identical interest stochastic games. But the convergence of our discrete time FP procedure in zero-sum game will be deduced from the continuous time system, thanks to an elaborate stochastic approximation technique. 


\paragraph{\ql{}} \citet{watkinsLearningDelayedRewards1989} introduced the \ql{} algorithm designed to control MDPs. It had a major impact and there are multiple generalizations, including offline \ql{} \cite{NEURIPS2020_0d2b2061}, double \ql{} \cite{NIPS2010_091d584f} or \ql{} with no-regret procedures \cite{kashCombiningNoregretQlearning2020}. There is a wide range of applications, from robots control \cite{taiRobotExplorationStrategy2016} to SAT solving \cite{kurinCanQlearningGraph2020}. \ql{} is a model-free algorithm, meaning that it does not require a complete specification of the environment such as the transition probability between states. A step proceeds as follows: starting from a state $s_t$, an action $a_t$ is chosen and this results in a new (random) state $s_{t+1}$ chosen by the environment while the learner gets an instantaneous payoff $R_{t+1}$. At every step, a $Q$-function $Q_t$ defined on every state-action pair is updated towards $R_{t+1}+\delta \max_a Q_t(s_{t+1}, a)$.

$Q$-learning was generalized to multi-agent systems. One line of work comprises algorithms that solve at every step the stage game defined as follows: every player has actions of the current state, and payoffs are the payoff of the $Q$-function $Q_t(s_t, \cdot)$. Then the values of the $Q$-function are updated towards the values of the stage game. This leads to algorithms such as Nash-Q \cite{huMultiagentReinforcementLearning1998, huNashQlearningGeneralsum2003} or Team-Q and Minimax-Q \cite{littmanValuefunctionReinforcementLearning2001}. For a complete survey, see \cite{busoniuComprehensiveSurveyMultiagent2008} and references therein. 


\paragraph{Combining $Q$-learning and fictitious play} To extend FP to stochastic games, the challenge is to define and compute what is a best-response to empirical observations: given a strategy for every player, the total discounted payoff is not straightforward to compute and is non-linear with respect a player stationary strategy. To overcome this difficulty, \citet{sayinFictitiousPlayZerosum2020} and \citet{leslieBestresponseDynamicsZerosum2020} use (close but different) mechanisms similar to that of \ql{} to deal with multiple states: a $Q$-function (or a state-value function) defined on every state-action pair or on every state is updated during the play. The player can then consider a stage game that is built with this $Q$-function, which is linear with respect to its mixed actions, to play a best response. The $Q$-function is typically updated at a slower timescale. More precisely, the algorithm of \citet{sayinFictitiousPlayZerosum2020} estimates $\hat Q_{i,s, k}(a)$ for every player $i$, state $s$, action $a$ at time $k$. It is the expected payoff if players play action profile $a$ starting from state $s$. Then, the procedure is to play the best-response against the belief on actions used by other players in current state, that is an element of $\argmaxi{a \in A^i} \hat Q_{i, s_k, k}(a, x_{s_k}^{-i})$ at time $k$ where $x_s^{-i}$ is the (uncorrelated) strategy of other players in state $s$ believed by player~$i$.

\citet{leslieBestresponseDynamicsZerosum2020} introduced and studied a best-response dynamics in zero-sum stochastic games. Time is continuous, so this is not an online-learning algorithm. The proposed dynamics maintains a vector  $u^i := \{u^i_s\}_{s \in \states}$ for player $i$. It is the expected payoff starting from every state $s$ (i.e., an estimate of the state-value function). It plays a role similar to that of the $Q$-function in $Q$-learning. Then player $i$ plays a best-response to the stage game with payoffs composed of the instantaneous payoff and the expected later payoff, that is an element of $\argmaxi{a \in A^i} (1-\delta)r^i_{s}(a, x_{s}^{-i}) + \delta P_s(a, x_{s}^{-i}) \cdot u^i$.

\cite{leslieBestresponseDynamicsZerosum2020} are only concerned with two-player zero-sum stochastic games and only study a continuous-time system. We introduce a family of continuous-time dynamics that contains the dynamics of \citet{leslieBestresponseDynamicsZerosum2020} as a particular case, as well as their discrete-time counterpart. We prove that our systems in both time regimes converge to stationary equilibria in identical interest and in zero-sum discounted stochastic games. 


%% file: include/fp_identical_interest.tex
\noicml
One important consequence of the convergence of the best-response dynamics in identical interest stochastic games is that it makes it possible to study a discrete-time fictitious-play like procedure for these games. In particular, we are going to define an asynchronous procedure that can be used as an online control algorithm, which is a novelty compared to \cite{leslieBestresponseDynamicsZerosum2020}.
\fi

\hypertarget{procedures-to-play-stochastic-games}{%
\paragraph{Procedures to play stochastic
games}\label{procedures-to-play-stochastic-games}}

Each asynchronous FP procedure, introduced in this subsection, is a behavioral strategy to play stochastic games for a player \(i\), that is a function
that provides a distribution of probability for the action \(a^i_{n}\)
given the history of the play prior to \(n\) and the current state
\(s_n\). Formally, it is a mapping
\(\bigcup_{n\in \mathbb N} [ (\states \times A)^n \times S ] \rightarrow \Delta(A^i)\).
Such an extension of FP is called \emph{asynchronous} because
there is a unique current state (that every player observes) and
actions are chosen by the players only for this state. This
contrasts with another FP procedure that we also study and call \emph{synchronous} FP. This not a behavioral strategy because there is no specific
current state and players provide actions for all states at every stage,
i.e., its a mapping
\(\bigcup_{n\in\mathbb N} (A^S)^n \rightarrow \Delta(A^i)^S\). It can be interpreted in various ways: either as a setting where the actual state is not known to the players, as a simulation of the system or as an algorithm, as in \cite{vriezeFictitiousPlayApplied1982}. 

\noicml Similarly to the best-response dynamics, our\fi \foricml Our \fi discrete-time procedures are designed using two estimates per state:
one is the empirical action that every player uses and the other one is the expected continuation payoff that a player
estimates starting from this state. First, we define the two estimates, then proceed with a description
of the action selection, and finally the updating rules.

\hypertarget{empirical-actions}{%
\paragraph{Empirical actions}\label{empirical-actions}}

We begin by exposing how the empirical action is computed for every
state. Given a state \(s \in S\) and a time step \(n\), \(s^\sharp_n\)
denotes the number of times that \(s\) occurs between \(0\) and \(n\)
i.e., \(s^\sharp_n = \sharp\{k\ |\ 0\leq k \leq n \land s_k = s\}\).
Then the empirical action of player \(i\) in state \(s\) is defined in
\(\Delta(A^i)\) as:
\[\dxi {n+1} s := \frac{1}{\countsn} \sum\limits_{k=0}^n 1_{s_k = s}a^i_k = \frac{1_{s_n = s} a^i_n}{\countsn} + \frac{\countsnp \dxi n s }{\countsn}\]
with the convention that if \(\countsn = 0\), then
\(\dxinp s= \dxi 0 s\) which is defined arbitrarily. Consequently, \(\dxi {n+1} s\) is equal to
\(\dxi n s\) when \(s_n\) is not equal to \(s\). Pure action $a^i_k$ is seen as an element of the Euclidean space $\Delta(A^i)$.

\paragraph{The auxiliary Shapley game} The second estimate is defined using the payoff of an auxiliary game. Given a continuation payoff vector $u \in \mathbb R^\states$, we define (following \cite{shapleyStochasticGames1953}) the auxiliary game parameterized by $u$ as the one-shot game where the action set is $\actions$ for every player $i$ and the payoff function is $\stfis{u^i}{\cdot}$ where:
\begin{equation*}
  \stfis{u}{x} := (1-\delta)r^i_s(x) + \delta \sum_{s' \in \states}P_{ss'}(x)u_{s'}
\end{equation*}

\citeauthor{finkEquilibriumStochasticPerson1964} (and Shapley) proved that stationary equilibria are the fixed point of an operator based on this auxiliary game. 

\paragraph{Update steps} $\durate$ is the non-increasing sequence of positive update steps for the payoff estimates, $\sdurate = \sum_{k=0}^n \durate[k]$, and starting values \(u^i_{0, s}\) defined arbitrarily. We suppose that:
\begin{equation}\label{hyp:discr}
\begin{gathered}
 \sum_k \frac{\durate[k]}{\sdurate[k]} = \infty \\
0 < \durate \leq 1 \ \ \ \ \ \ \ \ 
\alpha_{n+1} \leq \alpha_n
\end{gathered}
\tag{H1}
\end{equation}

\paragraph{Payoff estimates} Players estimate the continuation payoff in a vector \(u^i_{n} \in \mathbb R^S\). Values of this vector are written \(\dui n s\) for state \(s\),
at step \(n\) for player \(i\). \if\isicml0 We also use the generic notation \(u^i\)
for a vector of expected payoffs (with corresponding notation \(u^i_s\)
for payoffs starting from \(s\) when the difference between \(s\) and
\(n\) is not ambiguous).\fi At every step \(n\), the estimator is defined
as:
\if\isicml0
\[\begin{aligned}\dui {n+1} s & := \frac{1}{\countsn}\sum\limits_{k=0}^n\indic {s_k} s\stfis {u^i_k}{x_{k,s}}\end{aligned}\]
with, as previously, starting values \(u^i_{0, s}\) defined arbitrarily.
\else
\begin{equation*}
  \begin{aligned}
    \dui{n+1} s := & \frac{1}{\sdurate}\sum_{k=0}^n\durate[k]\stfis{u^i_k}{x_{k, s}} \\
    = & \frac{\sdurate[n-1]}{\sdurate} \dui n s + \frac{\durate}{\sdurate}\stfis{u^i_n}{x_{n, s}}
  \end{aligned}
\end{equation*}
Notice that in an identical interest stochastic game, $r^i_s$ does not depend on $i$ and as a consequence, $u^i_{n, s}$ does not depend on $i$ either. The same holds for zero-sum games because the payoff of player 2 is the negative of that of player 1 and it is sufficient to follow player 1's payoff. As such, we omit the superscript for $u^i$ in the rest of the paper.
\fi

Estimator \(\mdus\) can be seen as a mean where recent values of the
expected payoffs \(\mdfsu\) are given less weight than
oldest values. However, if the sequence \(\mdfsu\) is
stationary, \(\mdus\) will ultimately converge to the same limit
as \(\mdfsu\). A similar idea of a fast and a slow update rates
is used in \citep{leslieBestresponseDynamicsZerosum2020,perkinsAdvancedStochasticApproximation2013, kondaActorCriticTypeLearning1999, sayinFictitiousPlayZerosum2020}.

\paragraph{Remark:} If $\durate = 1$, then $\sdurate = n$.  Thus, the update rates of $x^i_{n, s}$ and $\mdus$ are the same and the hypothesis (H1) holds. It also holds for $\durate = \frac{1}{\log n}$ where $\sdurate = \log \log n$. In this case, the update rate of $x^i_{n,s}$ is much faster than the one of $\mdus$ and, as will be seen, this is the discrete-time analog of the continuous time dynamics in \cite{leslieBestresponseDynamicsZerosum2020}.

\hypertarget{action-selection}{%
\paragraph{Action selection}\label{action-selection}}

We can now define the action selection of our FP procedures. It is an extension of the classical
FP procedure. For
repeated games, FP is defined as a behavioral strategy
where at every stage, every player takes a best response against the empirical
action of the opponents up to that stage. For stochastic games, we define fictitious play as a best response in the auxiliary Shapley game parameterized by a given continuation payoff $u^i_n$, that is
for every \(n\):
\[a^i_{n,s} \in \bestr[i][s][\mdus[]][x^{-i}_{n, s}] := \argmax{y \in A^i} \mdfsu [y, x^{-i}_{n, s}]\]
where \(s=s_n\). When there are several best responses, our convergence results are independent on the selection rule.

Now we can define precisely our two FP procedures. 


\hypertarget{Asynchronous FP}{%
\paragraph{Asynchronous FP}\label{Asynchronous FP}}

\if\isicml0
\[\left\{\begin{aligned}
&\mdus[s][n+1]-\mdus  = \frac{\indic {s_n} s }{\countsn} \left( f^i_{s, u^i_n}(x_{n, s}) - \duin s\right) \\
&\dxinp s - \dxin s =  \frac{\indic {s_n} s }{\countsn} \left(a^i_{n, s}-\dxin s\right) \\ 
& a^i_{n} \in \bestr[i][s][u^i_n][x^{-i}_{n, s}]
\end{aligned}
\right.\tag{AFP}\label{eq:afp}\]
\else
\[\left\{\begin{aligned}
&\mdus[s][n+1]-\mdus  = \frac{\durate}{\sdurate} \left( \mdfsu - \mdus\right) \\
&\mdxsn[i][s][n+1] - \mdxsn[i] =  \frac{\indic {s_n} s }{\countsn} \left(a^i_{n}-\mdxsn[i]\right) \\ 
& a^i_{n} \in \bestr[i][s][\mdus[]][\mdxsn[-i]]
\end{aligned}
\right.\tag{AFP}\label{eq:safp}\]
\fi

This defines a behavioral strategy because only the current state is updated as the second equation shows. 

\hypertarget{Synchronous FP}{%
\paragraph{Synchronous FP}\label{Synchronous FP}}

To get convergence in non ergodic stochastic games, we now define a
version with synchronous updates on every state. This is an algorithm but not a behavioral strategy and thus, 
not a learning rule. \if\isicml0 The continuous
counterpart is simpler as it does not require an indicator function to
specify the current state.

If the stochastic game is played synchronously in every state (meaning
that for every state \(s\) and player \(i\) there is a choice of action
\(a^i_s \in A^i_s\) at every time step), the empirical action of player
\(i\) in state \(s\) is:
\[\dxinp s = \dfrac{\sum_{k=0}^n a^i_{k, s}}{n} = \dfrac{a^i_{n, s}}{n} + \dfrac{(n-1)\dxin s}{n}\]
and the estimated payoff starting from \(s\) is updated as follows:
\[\duinp s = \dfrac{\sum_{k=0}^n f^i_{s, u^i_k}(x_{k, s})}{n} =  \dfrac{f^i_{s, u^i_n}(x_{n, s})}{n} + \dfrac{(n-1)\duin s}{n}\]
These updates can be written in an incremental fashion, leading to synchronous fictitious play:
\begin{equation}\left\{\begin{aligned}&\duinp s-\duin s  = \frac{f^i_{s, u_n}(x_{n, s}) - \duin s}{n} \\& \dxinp s - \dxin s = \frac{a^i_{n, s}-\dxin s}{n} \\
\end{aligned}\right.\tag{SFP}\label{eq:sfp}\end{equation}

An alternative to both \ref{eq:sfp} and \ref{eq:afp} is semi-asynchronous fictitious play. It can be used during a standard asynchronous play of the stochastic game: $a^i_{n, s}$ are needed to update $x^i_{n, s}$ but not needed to update $u^i_{n,s}$, so only the updates on $x^i_{n, s}$ are asynchronous in SAFP:
\begin{equation}\left\{\begin{aligned}&\duinp s-\duin s  = \frac{f^i_{s, u_n}(x_{n, s}) - \duin s}{n} \\& \dxinp s - \dxin s = \indic {s_n} s\frac{a^i_{n}-\dxin s}{\countsn}\end{aligned}\right.\tag{SAFP}\label{eq:safp}\end{equation}

In identical interest games, as noted above, vectors
\(u^i_n\) are independent of player \(i\) as soon as players start with
the same belief, i.e., for all \(i, j \in I\), \(u^i_0 = u^j_0\), and
\(\stfis u x = \stfs u x\).\else

\[\left\{\begin{aligned}
  &\mdus[s][n+1]-\mdus  = \frac{\durate}{\sdurate} \left( \mdfsu - \mdus\right) \\
  &\mdxsn[i][s][n+1] - \mdxsn[i] =  \frac{1}{n} \left(a^i_{n}-\mdxsn[i]\right) \\ 
  & a^i_{n} \in \bestr[i][s][\mdus[]][\mdxsn[-i]]
  \end{aligned}
  \right.\tag{SFP}\label{eq:sfp}\]
\fi

\hypertarget{incremental-updates}{%
\paragraph{Incremental updates}\label{incremental-updates}}

In both FP procedures, \(\dxin s\) and \(\mdus\) can be computed \emph{via} incremental
updates. This will enable us to make the link with the continuous
version (Section~\ref{sec:link}). 
Moreover, it shows that for a machine implementation, the procedure
only needs constant memory instead of storing the whole history.


\begin{theorem}[Convergence of FP rules in identical interest stochastic games]
\label{thm:discrasync}

Under \ref{hyp:discr}, procedures \ref{eq:sfp} and \ref{eq:safp} almost surely converge to the set of stationary Nash
equilibria in
  identical interest ergodic discounted stochastic games\if\isicml1. \else
  and if \(\delta < 1/|S|\), then the results also hold for \ref{eq:afp}.\fi The convergence hold also in non ergodic games for \ref{eq:sfp}.

\end{theorem}

As in \cite{mondererPotentialGames1996}, our discrete-time proof for identical interest stochastic games is direct and is not derived from some associated continuous-time system. 

\if\isicml0

The proof of this latest theorem is sketched in the rest of the section. It uses the stochastic approximation framework. We recall the classical theory and the asynchronous extension that we need to modify.

\paragraph{Stochastic approximations}

Stochastic approximations have long been used to study asymptotic
behavior of discrete-time systems using their continuous-time
counterpart \citep{benaimDynamicalSystemApproach1996, benaimStochasticApproximationsDifferential2005, kondaActorCriticTypeLearning1999}.
In this framework, one typically assumes that there is a set valued function
\(F: \mathbb R^{K} \rightrightarrows \mathbb R^K\), a sequence of
decreasing update steps \(\{\gamma_n\} \in \mathbb R ^{\mathbb N}\) and
\(Y_{n+1}\) a noise difference random variable. Then the two
following systems are related:
\begin{equation}\frac{dy}{dt} \in F(y)\label{eq:approxcont}\end{equation}
\begin{equation}y_{n+1} - y_n - \gamma_{n+1}Y_{n+1} \in \gamma_{n+1} F(y_n)\label{eq:approxdiscr}\end{equation}

Stochastic approximations theorems (see Appendix~\ref{app:stochasticapprox} for a precise statement) typically assert that the limit set of a solution of (\ref{eq:approxdiscr}) is \emph{internally chain transitive} for differential inclusion (\ref{eq:approxcont}), meaning that two points of the limit set must be linked by a number of chained solutions of (\ref{eq:approxcont}):

\begin{definition}[Internally chain transitive]
  A set $A$ is internally chain transitive (ICT) for a differential inclusion $\diff \sax \in F(\sax)$ if it is compact and if for all $\sax, \sax' \in A$, $\epsilon > 0$ and $T > 0$ there exists an integer $n\in \mathbb N$, solutions $\sax_1, \ldots \sax_n$ to the differential inclusion and real numbers $t_1, t_2, \ldots, t_n$ greater than $T$ such that:
  \begin{itemize}
      \item $\sax_i(s) \in A$ for $0\leq s \leq t_i$
      \item $\|\sax_i(t_i)-\sax_{i+1}(0)\|\leq \epsilon$
      \item $\|\sax_1(0)-\sax\| \leq \epsilon$ and $\|\sax_n(t_n)-\sax'\|\leq \epsilon$
  \end{itemize}
  \label{def:ict_inline}
  \end{definition}

This framework links the synchronous systems: \ref{eq:sfp} and \ref{eq:sbrd} where $\alpha(t) = 1$ (see the next section for a discussion regarding other values of $\alpha$). In our case, $y = \{u_s, x_s\}_{s \in \states}$, there is no random noise and a vector of $F(y)$ is composed of an element of $\bestr i s {u} {x^{-i}_s} - x^i_s$ for lines corresponding to $x^i_s$ and $f_{s, u}(x_s) - u_s$ for lines corresponding to $u_s$. This proves the part related to \ref{eq:sfp} of Theorem~\ref{thm:discrasync}, see Appendix~\ref{app:stochasticapprox}.

\paragraph{Correlated Asynchronoucity} To do similar proofs for systems \ref{eq:safp} and~\ref{eq:afp}, one needs asynchronous stochastic approximations: the standard stochastic approximation framework \cite{benaimStochasticApproximationsDifferential2005}  is not sufficient to track every state and make every update rate depends on $\countsn$. We extend results from \citet{perkinsAsynchronousStochasticApproximation2012} to the case of correlated asynchronoucity. Indeed, in \ref{eq:safp}, variables $u_s$ are updated at every time step, independently of the current state and in \ref{eq:afp}, variables $u_s$ and $x_s$ are updated at the same time. Therefore, as \citet{perkinsAsynchronousStochasticApproximation2012} provided a theory where every variable was updated asynchronously with respect to the other ones. We extend this result, see Appendix~\ref{app:stochasticapprox}, and apply it to use systems \ref{eq:abrd} and \ref{eq:sabrd} to prove Theorem~\ref{thm:discrasync} under the ergodicity hypothesis.

\begin{proof}[Proof of Theorem~\ref{thm:discrasync} (sketch)]
  The first step of the proof is a generalization of the asynchronous stochastic approximation theorem of \cite{perkinsAsynchronousStochasticApproximation2012} to deal with correlated asynchronicity, this is done in Appendix~\ref{sec:fullproof}.

  Then, systems are written in the form of (\ref{eq:approxcont}) and (\ref{eq:approxdiscr}). This is explained above.

  Thus, it remains to show that the ICT sets of (\ref{eq:approxcont}) are contained in the set of stationary Nash equilibria and their associated payoff. To do this, we use results of convergence of the previous section. Howevbrefer, there is no direct implication between the convergence to a set and the fact that this set is ICT. Therefore, the chain transitivity is proven in part using the original definition, and in part with a Lyapunov function. More precisely we want to show that any ICT set is included in:

  \begin{align*}
    B := & \left\{(x, u) \left|
            \forall s \in \states\ \stfs u {x_s} \geq u_s\right.
          \right\} \\
   \text{\ and\ } A := & \left\{(x, u) \left|
            \begin{aligned}\forall s \in \states\ \forall i \in I, \ \stfs u {x_s} = u_s  \\ \land\  x^i_s \in \argmax{y^i \in A^i} \stfs u {y^i, x^{-i}_s}\\\end{aligned}\right.
          \right\}
  \end{align*}

  First, we show that any element $(x, u)$ of ICT sets are in $B$. Otherwise, we look at the chain between $(x, u)$ and itself, and conclude to a paradox: any solution is arbitrarily close to $B$ after a time $T$ independently of the starting point. Then, relatively to $B$, we can define a function $V(x, u):= \sum\limits_{s\in \states} \stfs u {x_s}$ which is a Lyapunov function in $B$. This makes it possible to conclude that any ICT set is included in $V^{-1}(0)$ which is equal to $A$.

  The whole proof is detailed in Appendix~\ref{app:stochasticapprox} of the supplementary material.

\end{proof}
\else

\begin{proof}[Proof of Theorem~\ref{thm:discrasync} (sketch)]
  The central idea is to show that the gap between $f_{s, u_n}(x_{n, s})$ and $u_{n, s}$ is lower-bounded by a sequence whose sum converge. This is possible because $f_{s, u_n}(x_{n, s})$ is mostly non-decreasing (but for the synchronization error of the players that optimize this function which is in $\frac{1}{n^2}$). This is similar to the proof continuous time (where the payoff function is a Lyapunov function of the system). Another key point is that $u_{n, s}$ moves towards $f_{s, u_n}$ at a rate no faster than the updates of $x_{n, s}$. This lower bound is used to prove the convergence of $u_{n, s}$, and the convergence to the set of stationary Nash equilibria follows.
  
  See Appendix~\ref{app:discreteproof} for the complete proof.
\end{proof}
\fi

%% file: include/best_response_dynamics.tex
This section extends and studies the best-response dynamics introduced and studied in zero-sum stochastic games by \citeauthor{leslieBestresponseDynamicsZerosum2020}. We  generalize their updating rates and prove that all the extended dynamics 
converge to stationary equilibria in identical interest and in zero-sum stochastic games. These dynamics are the continuous-time counterpart of \ref{eq:safp} and \ref{eq:sfp} as shown in the next section. 



As in discrete-time, there are two sets of variables:
\(\{u^i_s, x^i_s\}_{s \in S, i \in I}\). These variables may have different update rates, and we suppose
there is a function \(\alpha: \mathbb R^+ \rightarrow \mathbb R^{+*}\) to express the update rates of variables $u^i_s$. Function $\alpha$
is continuous and non-increasing. We make the following additional
assumption on \(\alpha\):
\begin{equation}\begin{aligned}
    \int_{0}^t \crateu y dy \limt +\infty\label{hyp:cont}\\
    \crateu t \geq 0 \text{\ and $\alpha$ is non-increasing}
    \end{aligned}\tag{H2}
\end{equation}

\noicml
\hypertarget{auxiliary-game}{%
\paragraph{Auxiliary game}\label{auxiliary-game}}

Following \citet{leslieBestresponseDynamicsZerosum2020} and previous
authors (for instance \citet{shapleyStochasticGames1953}), we define a
so-called \emph{auxiliary game} for every state \(s\) as a one-shot game
parameterized by a vector \(u := \{u_{s'}\}_{{s'}\in \states}\) whose
action set is \(A\) and payoff of player \(i\) is, for any action
profile \(x_s \in A\),
\(\stfis {u} {x_s} = (1-\delta)\ris(x_s) + \delta\sum_{s'\in S}\transSS(x_s)u_{s'}\).
Vector \(u\) represents the continuation payoff believed by player \(i\)
when it chooses an action, i.e., the expected payoffs starting from
every state, hence the term \(\delta \sum_{s'\in S}\transSS (x)u_{s'}\)
which is the discounted, expected payoff if players take action \(x\).

\paragraph{Synchronicity and asynchronicity}
In our paper, we study three types of systems: synchronous, fully asynchronous and semi-asyn\-chronous ones. In the synchronous kind, variables of all states are updated at the same time. There is no distinguished, current state. In semi-asynchronous systems, there is a current state, variables $x^i_s$ are updated only if they are related to the current states but variables $u^i_s$ are updated even if the current state is not $s$. In fully asynchronous systems, variables $x^i_s$ and $u^i_s$ are updated if and only if the current state is $s$.
\fi

\paragraph{Synchronous Dynamics} As in AFP, in the next dynamics, variables of all states are updated at the same time. 

For \(t \geq 0\) and every state \(s\) and player \(i\), synchronous
best-reply dynamics (SBRD) is defined as:
\begin{equation}\left\{ \begin{aligned} &\dot{u}^i_s(t) = \crateu t\left(\stfis{u^i(t)}{x_s(t)}-u^i_s(t)\right) \\ &\dot{x}_s^i(t) \in \bestr[i][s][u^i(t)][\cxst] -x^i_s(t) \end{aligned}\right.\tag{SBRD}\label{eq:sbrd}\end{equation}
where \(\bestr[i][s][u^i(t)][\cxst] := \text{argmax}_{a \in A^i} \stfis {u^i(t)} {a, x^{-i}_s(t)}\) (i.e., it is a best response to the auxiliary Shapley game). This action is used as an element of the Euclidean space $\Delta(A^i)$.  Vector \(u^i(t)\) denotes \(\{u^i_s(t)\}_{s\in S}\).

\textbf{Remark:} This is a generalization of the definition of \citeauthor{leslieBestresponseDynamicsZerosum2020} who studied the case $\alpha(t) = \frac{1}{t+1}$. Replacing $\stfis {u^i(t)}{x_s(t)}$ by the maximum over actions, that is $\max_{a \in A^i} \stfis {u^i(t)} {a, x^{-i}_s(t)}$ is an alternative that would be closer to the system outlined by \citeauthor{sayinFictitiousPlayZerosum2020} and \ql{} in general. It could be an interesting system to study but as noted by \citeauthor{sayinFictitiousPlayZerosum2020}, this would result in $u^i_s(t)$ to be different for two players even if the game is zero-sum or identical interest, which poses more theoretical challenges.

Differential inclusion \ref{eq:sbrd} classically admits a (typically
non-unique) solution
\citep{aubinDifferentialInclusionsSetValued1984, benaimStochasticApproximationsDifferential2005}. Indeed,
one can rewrite it as \(\frac{dy}{dt}\in F(t, y)\) where \(y\) is a
vector with every \(u^i_s, x^i_s\) and \(F\) is a closed set-valued map,
with non-empty, convex values. Furthermore, as 
shown in Lemma
\ref{lem:bounded} of Appendix \ref{app:continuous_proof}, values are bounded, so the solution is
defined on \(\mathbb R^+\)
\citep[p.~97]{aubinDifferentialInclusionsSetValued1984}.

In identical interest games, \(r^i_s = r_s\) for
every player \(i\). Therefore, for every $s$, \(u^i_s\) and
\(f^i_{s, u^i}\) do not depend on \(i\) (when initial values are equal), hence we omit the superscript
\(i\) in our statements. It is similar for zero-sum games.

\paragraph{Asynchronous Dynamics}
\noicml
We now provide results regarding the convergence of
semi-asynchronous systems and fully-asynchronous ones. In the semi-asynchronous one, the expected payoff starting from state~$s$ is always updated at the same rate but the empirical action is not, and for the fully asynchronous system, both the payoff estimates and the empirical action are updated at the same state-dependent rate.

The fully asynchronous system is defined as follows:
\begin{equation}\left\{ \begin{array}{l} \dot{u}_s(t) = \ratest s t\dfrac{\stfs {\cut} {\cxs {} t} -u_s(t)}{\crateu {\int_0^t \ratest s y dy}} \\ \dot{x}_s^i(t) \in \ratest st \left( \bestr[i][s][\cut][\cxs {-i} t]-\cxs i t\right) \\ \ratest s t \in [\ratelb, 1]\end{array}\right.\tag{ABRD}\label{eq:abrd}\end{equation}
where \(\beta_{-} \in (0, 1]\).

The semi-asynchronous system (in the spirit of the system of \citet{leslieBestresponseDynamicsZerosum2020}) is:

\else
We now provide results regarding the convergence of
asynchronous systems. In this system, the expected continuation payoff starting from state~$s$ is always updated at the same rate but the empirical action is not.
It is defined as follows:
\fi
\begin{equation}\left\{ \begin{array}{l} \dot{u}_s(t) = \crateu t \left(\stfs {\cut} {\cxs {} t} -u_s(t)\right) \\ \dot{x}_s^i(t) \in \ratest st \left( \bestr[i][s][\cut][\cxs {-i} t]-\cxs i t\right)  \\ \ratest s t \in [\ratelb, 1]\end{array}\right.\tag{ABRD}\label{eq:sabrd}\end{equation}
\foricml where \(\beta_{-} \in (0, 1]\).\fi

Value $\ratest s t$ is the update rate for state $s$ at time $t$. If only one state was updated at every time point, then we would have $\ratest s t$ equal to 0 but in one state where it would be equal to 1. 
If the game is ergodic, then on average every state is reached a strictly positive proportion of the time, $> \beta_-$. Next section will show in the ergodic case, that this system is formally linked to the AFP procedure.  

\begin{theorem}[Convergence of \ref{eq:sabrd} and \ref{eq:sbrd} in identical interest stochastic games]
\label{thm:contasync}

Let \(\{u_s, \beta_s, x^i_s\}_{s \in S, i \in I}\) be a solution of
\ref{eq:sabrd}. Under \ref{hyp:cont}, there is \(\Phi \in \mathbb R ^ {|\states|}\)
such that:
\vspace{-10pt}
\begin{itemize}
\tightlist
\item
  for all \(s\), \(\stfs {\cut} {x_s(t)} \limt \Phi_s\) and $\cut \limt \Phi$
\item
  \(\Phi\) is a stationary Nash equilibrium payoff
\item
  \(\{x_s(t)\}_{s\in S}\) converges to the set of stationary Nash equilibria with payoff $\Phi$ 
\end{itemize}
\noicml
It also holds for solutions of \ref{eq:abrd} when
\(\delta < \frac{1}{|\states|}\).\fi\end{theorem}

A sketch of the proof is provided below. A comprehensive proof with technical lemmas is provided in Appendix~\ref{app:continuous_proof}.
\noicml We conjecture it also holds for
\ref{eq:abrd} when \(\delta \geq \frac{1}{|S|}\), see simulations in Section~\ref{sec:simulations}.
\else

\begin{proof}[Sketch of proof]

\if\isicml0

We define, for $s\in \states$:
\begin{align*}\payt s t := & \stfs {u(t)} {\cxs {} t}\\ \optgap i s t := & \max\limits_{y \in A^i} \stfs \cut {y, \cxs {-i} t} - \stfs \cut {\cxs {} t}
\\ = & \max\limits_{y \in A^i} \stfs \cut {y, \cxs {-i} t} - \payt s t\end{align*}

We are going to lower bound $\payt s t - u_s(t)$ for every $s$ so as the differential of $u_s$ is lower-bounded by an integrable function. This guarantees that, as $u_s$ is bounded (see Lemma~\ref{lem:bounded}), it converges. We will then show that for every player $i$, $\optgap i s t \rightarrow 0$ and finish the proof of the Theorems by showing convergence of $\payt s t$ and studying the limit set of $\cxs {} t$.

Let $s \in \states$. First, note that \(\optgap i s t \geq 0\). Function $\pay s$ is differentiable:
\begin{equation}\diff {\pay s} = \delta \sum_{s'} P_{s s'}(x_s)\dot u_{s'} + \ratest s t \sum_i \optgap i s t \label{eq:diffgamma}\end{equation} where $\ratest s t = 1$ for \ref{eq:sbrd} and is already defined for \ref{eq:sabrd} and \ref{eq:abrd}. See Lemma~\ref{lem:diffgamma} in the supplementary material for details.

\mparagraph{Lower bound of  $\payt s t - u_s(t)$ for \ref{eq:sbrd} and \ref{eq:sabrd}} The lower bound of $\payt s t - u_s(t)$ is proven separately for \ref{eq:sbrd} and \ref{eq:sabrd} on the one hand and \ref{eq:abrd} on the other hand. Until further notice, we suppose that $\{u_s, x_s\}_{s \in \states}$ is \textbf{a solution of \ref{eq:sabrd}} (which includes the case \ref{eq:sbrd}).

Let
\(\wstate t \in \argmin{s' \in S} \left(\payt {s'} t - u_{s'}(t)\right)\). Then, for any $s \in \states$:

\begin{equation}\begin{aligned}\frac{d \pay s}{dt} & \geq \delta \sum_{s'}  \transSS(x_s)\frac{\payt {s'} t-u_{s'}(t)}{\crateu t} \\ & \geq \delta \sum_{s'}P_{ss'}(x_s) \frac{\payt {\wstate t} t -u_{\wstate t}}{\crateu t} \\ & = \delta  \frac{\payt {\wstate t} t-u_{\wstate t}}{\crateu t}\end{aligned}\label{eq:bound_gamma_si}\end{equation}

Moreover, for \(h > 0\):
\begin{equation}\begin{aligned}&\payt {\wstate {t+h}} {t+h} -u_{\wstate {t+h}}(t+h)- (\payt {\wstate t} t - u_{\wstate t}(t)) \\&  \geq \payt {\wstate {t+h}} {t+h} -u_{\wstate {t+h}}(t+h) - (\payt {\wstate {t+h}} t - u_{\wstate {t+h}}(t)) \\ & \geq h \min_{s'\in S} \frac{d\pay {s'}}{dt} + o(h)  + u_{\wstate{t+h}}(t)-u_{\wstate{t+h}}(t+h)\end{aligned}\label{eq:pregronwalli}\end{equation}

Then, it can be shown that if $s'$ is an accumulation point of $\wstate {t+h}$ when $h$ goes to $0$, then:

\begin{equation}u_{s'}(t)-u_{s'}(t+h) = -h\frac{\Gamma_{s_-(t)}(t)-u_{s_-(t)}(t)}{\alpha(t)} + o(h)\label{eq:pregronwal_ui}\end{equation}

Since this is valid for every such $s'$, combining (\ref{eq:bound_gamma_si}), (\ref{eq:pregronwalli}) and (\ref{eq:pregronwal_ui}) gives:

\[\begin{split}\payt {\wstate {t+h}} {t+h} -u_{\wstate {t+h}}(t+h) - (\payt {\wstate t} t - u_{\wstate t}(t)) \optnl \geq h(\delta-1) \frac{\Gamma_{s_-(t)}(t)-u_{s_-(t)}}{\alpha(t)}+o(h)\end{split}\]

Now we are going to apply a version of Grönwall Lemma (see details in Lemma~\ref{lem:contconv}) so we get:

\optmult{
  \payt {\wstate t} t - u_{\wstate t}(t) \optnl \geq \left( \payt {\wstate 0} 0 - u_{\wstate 0}(0)\right)\exp\left(\int_0^t (\delta-1)\frac{1}{\crateu t}dt\right)
}

Therefore, we can use this inequality in $\dot u_s(t)$:
$$\dot u_s(t) \geq - \frac{A}{\crateu t} \exp\left(\int_1^t (\delta-1) \frac{1}{\crateu t}dt\right)$$ where $A > 0$. The right hand side term is integrable, and as $u_s$ is bounded (see Lemma~\ref{lem:bounded}), it converges.

\mparagraph{Lower bound of\ $\payt s t - u_s(t)$ for \ref{eq:abrd}} We now consider a solution of system \ref{eq:abrd}. A similar reasoning can be done with function $\paysup(t) := \sum_{s\in \states} \pos{u_s(t)-\payt s t}$, see Lemma~\ref{lem:convabrd} for details. We can bound $\paysup$ and in the end we can also bound $\dot u_s$
\begin{align*}
  \dot u_s(t) & \geq -A \frac{\exp\left(\int_0^t (\delta|S|-1)\frac{\beta_-}{\crateu {\int_0^v \ratest s w dw}}dv\right)}{\crateu {\int_0^t \ratest s v dv}}
\end{align*}
which is integrable. Therefore, $u_s$ converges with the same arguments.

\mparagraph{$\sum_{i\in I}\optgap i s t$ goes to $0$}
In either case, we show that $\optgap{i}{s}{t}\rightarrow 0$. First, we notice that in the \ref{eq:sbrd} and \ref{eq:sabrd} case:

\begin{align*}\int_0^t \sum_{i\in I}\optgap{i}{s}{v} dv & \leq  \int_0^t \frac{\ratest s u}{\beta_-}\sum_{i\in I}\optgap{i}{s}{v}dv \\ & =  \frac{1}{\beta_-}\left(\int_0^t \diff {\pay s}(v) - \delta \sum_{s'} P_{s s'}(x_s(v))\dot u_{s'}(v)dv\right)
\\ &\leq \frac{1}{\beta_-} \left(\payt s t - \payt s 0\right) \optnl \optesp \optspace{\ \ \ } + \frac{A}{\beta_-(1-\delta)}  \left(1 -\exp\left(\int_1^t \frac{1-\delta}{\crateu v}dv\right)\right)\end{align*}

So, this integral is bounded. However, as $\sum_{i\in I}\optgap i s \cdot$ is Lipschitz (see Lemma~\ref{lem:deltalip}), we conclude that $\sum_{i\in I}\optgap i s t \limt 0$ (Lemma~\ref{lem:delta0}).

In the \ref{eq:abrd} case, we have similar inequalities except for the argument of $\alpha$ which is $\int_0^t \ratest s v dv$. As it is bounded between $\beta_- t$ and $t$, it does not change much of the computations and the integral is bounded as well. See Lemma~\ref{lem:delta0async} for details.

\mparagraph{Convergence of\ \ $\pay s$} In the case of \ref{eq:sbrd} and \ref{eq:sabrd}, using (\ref{eq:diffgamma}), we can lower bound its derivative:

\begin{align*}\diff {\pay s } \geq & \delta \frac{\payt {\wstate t} t - u_{\wstate t}(t)}{\crateu{t}} \optnl \geq \optesp - \delta A \frac{\exp\left(\int_0^t (\delta-1)\frac{1}{\crateu u}du\right)}{\crateu{t}}\end{align*}

This latest term being integrable and $\pay s$ being bounded, we conclude that $\pay s$ converges to its $\lim\sup$ when $t$ goes to $+\infty$. The limit is necessarily the same as $u_s$, otherwise $u_s$ could not be bounded (see the comprehensive proof for details).

\mparagraph{Limit set of\ \ $\cxs {} t$} Let \(\tilde x\) be an accumulation point of the vector-valued function
\(x = \{x_s\}\). Then, we previously showed:
\[\optgap i s t = \stfs \cut {br^i_{s, u(t)}(x^{-i}_s(t)) - x^i_s(t), x^{-i}_s(t)} \limt 0\]
So by continuity, for all $s$:
\[f_{s, \lim u}(br^i_{s, \lim u}(\tilde x^{-i}) - \tilde x^i_s, \tilde x^{-i}_s) =0\]

So $\tilde x$ belongs to the set of Nash equilibria.

\else

We define, for $s\in \states$:
\begin{align*}\payt s t := & \stfs {u(t)} {\cxs {} t}\\ \optgap := & \max\limits_{y \in A^i} \stfs \cut {y, \cxs {-i} t} - \stfs \cut {\cxs {} t} \geq 0\end{align*}

We are going to lower bound $\payt s t - u_s(t)$ for every $s$ so as the differential of $u_s$ is lower-bounded by an integrable function. This guarantees that, as $u_s$ is bounded (see Lemma~\ref{lem:bounded}), it converges. We then show that for every player $i$, $\optgap \rightarrow 0$ and finish the proof of the theorems by studying convergence of $\payt s t$ and the limit set of $\cxs {} t$.

\fi
\end{proof}


\begin{theorem}[Convergence of \ref{eq:sabrd} in zero-sum stochastic games]
\label{thm:contzerosum}

Let \(\{u_s, \beta_s, x^i_s\}_{s \in S, i \in I}\) be a solution of
\ref{eq:sabrd}. There exists a constant $A > 0$ (which only depends on $\delta$ and $r_{s}$) such that if $\alpha^\star > \lim{t \rightarrow \infty} \crateu t$, then, under~\ref{hyp:cont}:
\vspace{-10pt}
\begin{itemize}
\tightlist
\item
  for all \(s\), \(\lim\sup_{t\rightarrow \infty} |\stfs {\cut} {x_s(t)} - \cut | \leq A \alpha^\star \)
\item
  \(\{x_s(t)\}_{s\in S}\) converges to the set of stationary Nash $A\alpha^\star-$equilibria as $t\rightarrow \infty$.
\end{itemize}
\end{theorem}

The proof is in Appendix~\ref{app:zerosum_continuous}. Note that if $\alpha(t)\rightarrow 0$, then $\alpha^\star$ can be chosen arbitrarily close to $0$ which is the case in \cite{leslieBestresponseDynamicsZerosum2020} ($\alpha(t)=t+1)$. Hence this is an extension of \cite{leslieBestresponseDynamicsZerosum2020}.


%% file: include/link_continuous_discrete.tex
This section uses the continuous time results of the previous section to deduce the convergence of the discrete-time procedures in zero-sum stochastic games (which is not covered by \cite{leslieBestresponseDynamicsZerosum2020}) but also in team stochastic games (defined below). Any identical interest stochastic game is a team stochastic game but not the reverse. 

\paragraph{Zero-sum games} We can discretize the continuous model using an extension of the stochastic approximation framework (see details in Appendix~\ref{app:stochasticapprox} and then using an algorithm with doubling trick (standard in RL procedures). Note that the doubling trick trigger $T(\alpha)$ can be computed, this is explained in Appendix~\ref{app:stochasticapprox}.

\begin{algorithm}
\begin{algorithmic}

\STATE $\alpha, x, u \leftarrow 1, x_0, u_0$

\LOOP
    \STATE $x_{n+1, s_n} \leftarrow x_{n, s_n} + \frac{1}{n+1}\left(a_n - x_{n, s_n}\right)$
    \STATE $\forall s, u_{n+1, s} \leftarrow u_{n, s} + \frac{\alpha}{\sigma_n}(f_{s, u_n}(x_{n, s}-u_{n, s})$
    \STATE Choose $a^i_{n+1} \in \text{br}^i_{s, u_n}(x^{-i}_{n+1, s})$
    \IF {$n > T(\alpha)$} \STATE $\alpha \leftarrow \alpha/2$
    \ENDIF
\ENDLOOP
\end{algorithmic}
\caption{FP with Doubling Trick for Zero-Sum Games}
\end{algorithm}

Combining the stochastic approximation framework and Theorem~\ref{thm:contzerosum} guarantees the convergence of this algorithm. A detailed proof is in Appendix~\ref{app:stochasticapprox}.

\vspace{-10pt}
\begin{theorem}[Convergence of FP with doubling trick in zero-sum stochastic games]
\label{thm:discrzs}

Under \ref{hyp:discr}, procedures \ref{eq:sfp} and \ref{eq:safp} with the doubling trick as specified in Algorithm 1 almost surely converge to the set of stationary Nash
equilibria in
  zero-sum ergodic discounted stochastic games. The convergence holds also in non ergodic games for \ref{eq:sfp}.

\end{theorem}

\paragraph{Different Priors and Team Stochastic Games}

\citet{hollerLearningDynamicsReinforcement2020} proved that in exact global-potential stochastic games, players are divided into two categories: either they do not influence the transition or they have the same payoff function up to a constant. This last class is called \emph{team stochastic games}. While the proof of convergence of FP in discrete-time is not straightforward in team games, it can be studied using the continuous time system and stochastic approximations techniques. Similarly, if players have different priors on continuations.

\begin{theorem}[Convergence of FP with different priors in team stochastic games]
If all players use a FP procedure as defined in Section~\ref{sec:fictitious_play} with priors on the continuations that may be different (i.e., $u^i_s(0)$ may not be equal to $u^j_s(0)$ for two players $i$ and $j$ and a state $s$), then the average actions $x_{s, n}$ and vectors $u^i$ (for every player $i$) converge respectively to the set of stationary Nash equilibria and the corresponding continuation payoffs.
\end{theorem}

\vspace{-10pt}\begin{proof}[Sketch of proof]
We can define a Lyapunov function on the continuous-time system, yielding conditions on chain transitive sets. Details can be found in \ref{sec:priors_team_games}.
\end{proof}

%% file: include/conclusion.tex
We defined a number of continuous and discrete time systems to learn stationary equilibria in stochastic games. They combine ideas from fictitious play and {\ql} and are extensions of a continuous-time system of \cite{leslieBestresponseDynamicsZerosum2020} who proved its convergence to stationary equilibria in zero-sum stochastic games. We prove their convergence to stationary equilibria  in continuous time but also in discrete time; in zero sum but also in identical interest discounted stochastic games. An open problem is to show the convergence of the procedures of \cite{sayinDecentralizedQLearningZerosum2021} in identical interest and team stochastic game. The main difficulty relies on the fact that their updating rule does not preserve the identical interest objective along the trajectory.

Another interesting direction is the speed of convergence. As outlined in the proof, there are bounds for zero-sum stochastic games but none for identical interest ones. To the best of our knowledge, no results are known even in non-stochastic games \cite{mondererFictitiousPlayProperty1996}.

An interesting extension would be limiting average stochastic games. This could be achieved by increasing in \ref{eq:safp} the discount factor $\delta_n$ from stage to stage to 1. Another nice extension would be to have a model free algorithm, that is an updating rule that does not use the knowledge of the probability distribution. A simple way to adapt \ref{eq:safp} is by letting the players explore with small probability and replace in the up-dating rule of \ref{eq:safp} the probability transition by its empirical estimation. When the game is ergodic, after some period $T(\varepsilon)$ the estimated transition probability will be close to the real transition with probability at least $1-\varepsilon$. Consequently, under ergodicity and (\ref{hyp:discr}), this modification of \ref{eq:safp} converges to stationary equilibria in identical-interest and zero-sum stochastic games. 

More tricky is to construct a learning procedure that converges to a stationary equilibrium when the players do not observe other players past actions but only their own past actions and the current state. Even when there is only one state (a repeated game) FP and all its variants fail because there is no way to form a belief about the opponents without observing their actions. A class of procedures that converge to Nash equilibria of the stage game in zero sum and identical interest repeated games are non-regret algorithms \cite{blumExternalInternalRegret2005,hofbauerGlobalConvergenceStochastic2002} with some exploration to be able to estimate what would the payoff be if a player has played differently (we are in a bandit setting). But regret is not well defined in stochastic games even if players observe the past actions of the opponents \cite{mannorEmpiricalBayesEnvelope2003}. 


A last interesting question is: what happens if we let each player use simple Q-learning? This is very unstable in some repeated games (RG) \cite{wunderClassesMultiagentQlearning2010}. In others RG, simulations in a repeated pricing game (a kind of repeated prisoner dilemma) show that Q-learning does not converge to the stationary Nash  \cite{calvanoArtificialIntelligenceAlgorithmic2020} (i.e., not competitive pricing which would be defection at every stage) but to Pareto Nash equilibrium of the RG (to collusion, that is a cooperative equilibrium, similar to Tit for Tat). One may wonder if it is possible to construct a  Q-learning like procedure which converges to Pareto optimal equilibria in every repeated/stochastic game. This was our motivating question because the Tit for Tat equilibrium in a RG is a stationary equilibrium in the auxiliary stochastic game where the current state is the last action profile in the RG.

%% file: include/appendix_direct_discrete_proof.tex

{

\input{include/notations_appendix_direct_discrete_proof}

In this section, we prove that systems \ref{eq:sfp} and \ref{eq:safp} converge to a stationary Nash equilibrium. The proofs for the two systems are similar, except for the last part about the convergence of the empirical actions. Therefore, we write the first, identical part, only for \ref{eq:safp} (the more complex system), and give the two proof in the last part.

Recall system \ref{eq:safp}:
\begin{equation}\left\{
    \begin{aligned}
        & \du[n+1] -\du  = \frac{\durate}{\sdurate}\left(\df - \du\right) \\
        & \dx[n+1][s][i] - \dx[n][s][i] = 1_{s = s_n}\frac{a^i_{n}-\dx[n][s][i]}{n} \\
        & a^i_{n} \in \bestr[i][s_n][\du[n][]][\dxn{-i}{s}] \\ 
        & \sdurate = \sum_{k=1}^n \durate[k]
    \end{aligned}\right.
    \if\isicml0
    \tag{SAFP}
    \else
    \tag{AFP}
    \fi
    \label{app:eq:safp}\end{equation}

Under these hypothesis on $\alpha_n$:
\begin{equation}\label{app:hyp:discr}
\begin{aligned}
\sum_k \frac{\durate[k]}{\sdurate[k]} = \infty \\
\durate \leq 1 \\
\alpha_{n+1} \leq \alpha_n
\end{aligned}
\tag{H1}
\end{equation}
Note that this is satisfied for $\durate = 1$ (single timescale, autonomous case) or $\durate = \frac{1}{n}$ ($u$ is updated significantly slower than $x$).

We are going to show that under \ref{app:hyp:discr}, $\du$ and $\df$ converge to the same equilibrium payoff for every $s$ in an (ergodic for \ref{eq:safp}), identical interest stochastic game. This implies that $x_{n, s}$ converges to a stationary Nash equilibrium.

\begin{proof}

Denote $\fgap := \df$, \ $\dw := \min_{s \in \states} \fgap - \du$ and $\ds \in \argmin{s \in \states}\, \fgap - \du$ so as $\dw = \fgap[\ds] - \du[n][\ds]$. 

\paragraph{The energy of the system $\dw$ converges} We bound the changes in $\dw$:
\begin{equation}\begin{aligned}
    \dw[n+1] - \dw & = \fgap[\ds[n+1]][n+1] - \du[n+1][\ds[n+1]] - (\fgap[\ds] - \du[n][\ds]) \\
    & \geq \fgap[\ds[n+1]][n+1] - \du[n+1][\ds[n+1]] - (\fgap[\ds[n+1]] - \du[n][\ds[n+1]]) \\
    & = \fgap[\ds[n+1]][n+1] - \fgap[\ds[n+1]] - (\du[n+1][\ds[n+1]] - \du[n][\ds[n+1]]) \\
    & = \fgap[\ds[n+1]][n+1] - \fgap[\ds[n+1]] - \frac{\durate[k]}{\sdurate[k]}\left(\fgap[\ds[n+1]] - \du[n][\ds[n+1]]\right) \\
\end{aligned}
\label{eq:5}
\end{equation}

Now, there exists $C$ (independent of $n$) such that $|\fgap[\ds[n+1]] - \du[n][\ds[n+1]]| - |\fgap[\ds] - \du[n][\ds]| < \frac{C}{n}$ (because for every $s$, $\fgap - \du$ changes of at most $\frac{C}{n}$ between $n$ and $n+1$, so this is true for the minimum as well).

As a consequence, continuing \eref{eq:5}:
$$
\begin{aligned}
    \dw[n+1] - \dw & \geq \fgap[\ds[n+1]][n+1] - \fgap[\ds[n+1]] - \frac{\durate\dw}{\sdurate} - \frac{\durate C}{n\sdurate} \\
    &\begin{multlined} \geq \df[\ds[n+1]][\du[n+1][]][\dx[n+1][\ds[n+1]]] - \df[\ds[n+1]][\du[n][]][\dx[n+1][\ds[n+1]]] \\ + \df[\ds[n+1]][\du[n][]][\dx[n+1][\ds[n+1]]] - \df[\ds[n+1]][\du[n][]][\dx[n][\ds[n+1]]] - \frac{\durate \dw}{\sdurate} - \frac{\durate C}{n\sdurate} \end{multlined} \\
    & \begin{multlined}\geq \delta \sum_{s'\in \states}\probatr[s'][\ds[n+1]][\dx[n+1][\ds[n+1]]]\left(\du[n+1][s'] - \du[n][s']\right) \\ + \df[\ds[n+1]][\du[n][]][\dx[n+1][\ds[n+1]]] - \df[\ds[n+1]][\du[n][]][\dx[n][\ds[n+1]]] - \frac{\durate \dw}{\sdurate} - \frac{\durate C}{n\sdurate} \end{multlined} \\ 
    & \geq \delta  \frac{\dw}{\sdurate \durate} + \df[\ds[n+1]][\du[n][]][\dx[n+1][\ds[n+1]]] - \df[\ds[n+1]][\du[n][]][\dx[n][\ds[n+1]]] - \frac{\durate \dw}{\sdurate} - \frac{\durate C}{n\sdurate } \\
    & \geq (\delta -1)  \frac{\durate \dw}{\sdurate} + \df[\ds[n+1]][\du[n][]][\dx[n+1][\ds[n+1]]] - \df[\ds[n+1]][\du[n][]][\dx[n][\ds[n+1]]] - \frac{\durate C}{n\sdurate}
\end{aligned}
$$

\newcommand{\fsf}{\df[\ds[n+1]][\du[n][]][\dx[n+1][\ds[n+1]]]}

The first order expansion of $\fsf$ for \ref{eq:safp}:
\begin{equation}
    \begin{aligned}
\df[\ds[n+1]][\du[n][]][\dx[n+1][\ds[n+1]]] = \df[\ds[n+1]][\du[n][]][\dx[n][\ds[n+1]]] + \sum_{i \in I}\frac{1_{s=s_n}}{n}\left(\df[\ds[n+1]][\du[n]][a^i_n, \dx[n][s][-i]] - \df[\ds[n+1]][\du[n]][\dx[n][s]]\right) + O\left(\frac{1}{n^2}\right)
\end{aligned}
\label{eq:fn}
\end{equation}

The expansion \eref{eq:fn} would be the same for \ref{eq:sfp} except for the indicator $1_{s=s_n}$ which would disappear.

In any case, the first order term is positive (because $a^i_n$ is a best-response in the auxiliary game), therefore there exists $D > 0$ such that:
$$
\begin{aligned}
    \dw[n+1] - \dw & \geq (\delta -1)  \frac{\durate \dw}{\sdurate}  - \frac{D}{n^2} - \frac{\durate C}{n\sdurate}
\end{aligned}
$$

Then, using Lemma~\ref{lem:discretegronwall}, for $n > m$:
$$\dw[n] \geq \dw[m] \Pi_{k=m}^n(1+\frac{\delta-1}{k}) - \sum_{k=m}^n \left[\frac{D}{k^2} + \frac{\durate[k] C}{k\sdurate[k]}\right] \geq E\Pi_{k=m}^n(1+\frac{\delta-1}{k}) -\sum_{k=m}^\infty  \left[ \frac{D}{k^2} + \frac{C}{k^2} \right]$$
for some $E > 0$ (independent of $m$ because $\dw$ is bounded). The last inequality is obtained using Lemma~\ref{lem:bounddiscralpha}. The right term goes to 0 as the rest of a convergent sum. Furthermore, the left term goes to $0$ when $n$ goes to $\infty$, so $\dw \rightarrow 0$.

\paragraph{The continuation payoffs $u_n$ converge}
Sequence $\dw$ can be lower bounded more precisely: $\sum_{k=m}^\infty \frac{D+C}{k^2} = \Omega \left( \frac{1}{m} \right)$ and $\Pi_{k=m}^n(1+\frac{\delta-1}{k}) = \Omega\left(\left(\frac{m}{n}\right)^{\delta-1}\right)$, so with $m = [\sqrt n]$, $w_n \geq \Omega\left(\frac{1}{\sqrt n}\right) + \Omega \left(\frac{1}{n^{\frac{1-\delta}{2}}}\right) = \Omega \left(\frac{1}{n^{\frac{1-\delta}{2}}}\right)$.

Consequently, for every $s$, $\du[n+1]-\du \geq \Omega\left(\frac{1}{n^{1+\frac{1-\delta}{2}}}\right)$, and as $\du$ is bounded (again using Lemma~\ref{lem:bounddiscralpha}), it converges.

\paragraph{The payoff of the auxiliary game converges to the same limit}
Similarly, one can show that $\df[s][\du[n+1][]][\dx[n+1]] - \df$ is lower bounded by $\Omega\left(\frac{1}{n^{1+\frac{1-\delta}{2}}}\right)$, so it converges, and it is the same limit as $\du$ (otherwise $\du$ could not be bounded).

\paragraph{The limit is an equilibrium payoff in \ref{eq:safp}} Using (\ref{eq:fn}), (valid for every $s$), writing $\Delta_{s, n} := \sum_{i \in I}\df[s][\du[n]][a^i_n, \dx[n][s][-i]] - \df[s][\du[n]][\dx[n][s]]$:
\begin{equation}
    \begin{aligned}
\df[s][\du[n][]][\dx[n+1][s]] - \df[s][\du[n][]][\dx[n][s]]= \frac{1_{s=s_n}}{n}\Delta_{s, n} + O\left(\frac{1}{n^2}\right)
\end{aligned}
\label{eq:1}
\end{equation}

Then:
\begin{equation}
\begin{aligned}
\df[s][\du[n+1][]][\dx[n+1][s]] - \df[s][\du[n][]][\dx[n+1][s]] & = \delta \sum_{s'\in \states} P_{ss'}(\dx[n+1][s]) (\du[n+1][s'] - \du[n][s']) \\
& = \delta \sum_{s'\in \states} P_{ss'}(\dx[n+1][s]) \frac{\durate}{\sdurate} (\df[s'] - \du[n][s']) \\
& \geq \delta |S| P_{ss'}(\dx[n+1][s]) \frac{\durate}{\sdurate}  \dw
\end{aligned}
\label{eq:2}
\end{equation}

Summing (\ref{eq:1}) and (\ref{eq:2}) gives:
\begin{equation}
    \df[s][\du[n+1][]][\dx[n+1][s]] - \df[s][\du[n][]] \geq \delta |S| \frac{\durate}{\sdurate} \dw + \frac{1_{s=s_n}}{n}\Delta_{s, n} + O\left(\frac{1}{n^2}\right)
\label{eq:3}
\end{equation}

However, $\Delta_{s, n} \geq 0$, so summing (\ref{eq:3}) over $n$ gives that $\sum_n \frac{1_{s=s_n}}{n}\Delta_{s, n} < \infty$ with the same reasoning as above (because $\frac{\durate}{\sdurate} \dw = \Omega\left(\frac{1}{n^{1+(1-\delta)/2}}\right)$ and the terms in the left term cancel out).

Simple calculations yield that $\frac{1}{n}\sum_{k=1}^n 1_{s = s_k} \Delta_{s, k}$ goes to $0$ as $n$ goes to $\infty$. However, it is clear that changes in $\Delta_{s, n}$ are of the order of magnitude of the update steps, that is $\frac{1}{n}$. As a consequence, assuming that $\Delta_{s, n}$ does not go to $0$, there exists a $A > 0$ such that for $\epsilon > 0$, if $\Delta_{s, n} \geq 3\epsilon$, then $\Delta_{s, n+m} \geq 3\epsilon - \sum_{k=1}^m\frac{A}{n+k} \geq 2\epsilon - A \log((n+m)/n)$ for $n$ large enough (well known result of the harmonic series). But then, for $m=[n\exp(\epsilon/A)-1)]$:

\begin{equation*}
    \frac{1}{n+m}\sum_{k=n}^{m+m} 1_{s = s_k} \Delta_{s, k} \geq \frac{1}{n+m} \sum_{k=n}^{n+m} 1_{s=s_k} \epsilon
\end{equation*}

Since the game is ergodic, with probability $1$ when $m$ goes to $\infty$ (and it goes to infinity when $n$ goes to infinity), $\frac{1}{m}\sum_{k=n}^{n+m} 1_{s=s_k}$ is greater than a real $\beta_-$  which depends only on the game (minimal frequency of visit of $s$ using the law of large numbers).

\begin{equation*}
    \frac{1}{n+m}\sum_{k=n}^{m+m} 1_{s = s_k} \Delta_{s, k} \geq \frac{1}{n+m} m \beta_- \epsilon \geq \frac{n\exp(\epsilon/A)-1)-1}{n+ n\exp(\epsilon/A)-1)} \beta_- \epsilon \geq \frac{\exp(\epsilon/A)-1) - \frac{1}{n}}{1 + \exp(\epsilon/A)-1)}\beta_- \epsilon
\end{equation*}

This latest inequality is absurd, so $\Delta_{s, n}$ goes to $0$ almost surely when $n$ goes to infinity, proving that the limit of $u_{n, s}$ is an equilibrium payoff. Then, it is clear that $x_{n, s}$ converge towards the set of Nash equilibria almost surely (otherwise $f_{s, u_{n, s}}(x_{n, s})$ could not have the same limit as$u_{n, s}$.

\textbf{Now, for \ref{eq:sfp}}, the proof is similar but for the indicator function which disappears. As a consequence, it is not needed to use a $\beta_-$ but we still prove that almost surely, $\Delta_{s, k}$ goes to $0$. Therefore, we do not need the ergodicity hypothesis for this case.

\end{proof}
}

%% file: include/notations_appendix_direct_discrete_proof.tex
\NewDocumentCommand{\du}{O{n} O{s} O{}}{u^{#3}_{#1\ifthenelse{\equal{#2}{}}{}{,#2}}}
\NewDocumentCommand{\dx}{O{n} O{s} O{}}{x^{#3}_{#1, #2}}
\NewDocumentCommand{\df}{O{s} O{\du[n][]} O{\dx}}{f_{#1, #2} \left(#3\right)}
\NewDocumentCommand{\dw}{O{n}}{w_{#1}}
\NewDocumentCommand{\ds}{O{n}}{s^-_{#1}}

%% file: include/appendix_continuous.tex

In this section, we prove that the best-response dynamics converge in identical interest stochastic games. We first prove a boundedness lemma and then proceed with the convergence of every system in identical interest stochastic games.

Note that since we only deal with identical interest stochastic games, the superscript $i$ in $u^i_s$ can be omitted as all $u^i_s$ are equals (see Section~\ref{sec:brd}).

In what follows, let \(\{u^i_s, x^i_s\}_{s\in \states, i \in I}\) be a solution of
(\ref{eq:sabrd}). Note that since (\ref{eq:sbrd}) is included in (\ref{eq:sabrd}), this is valid for solutions of (\ref{eq:sbrd}) as well (it is the case where $\beta_- = 1$).

We define:
\[\begin{aligned}\payt s t := &\stfs \cut {\cxs {} t} \\ \optgap := & \max\limits_{y \in A^i} \stfs \cut {y, \cxs {-i} t} - \stfs \cut {\cxs {} t} \\ = & \max\limits_{y \in A^i} \stfs \cut {y, \cxs {-i} t} - \payt s t\end{aligned}\]

\begin{lemma}[] \label{lem:bounded}

Let \(\{u^i_s, x^i_s\}_{s\in \states, i \in I}\) a solution of
(\ref{eq:sabrd}) or (\ref{eq:sbrd}). Then for all \(s \in \states\), functions \(u_s\) and
\(t \mapsto \stfs \cut \cxst\) are bounded.\end{lemma}

\begin{proof}

Let
\(M = \max_{s\in S, a \in A} \left\{|u_s(\startt)|, |\payt s \startt|, |r_s(a)|\right\}+1\).

Then \(|u_s(\startt)| < M\) and \(|\payt s \startt| < M\) for every
\(s\). \(u_s\) and \(\pay s\) are continuous, therefore if they are not
bounded by \(M\), there exists \(t\) minimal such that there exists
\(s \in S\) such that either:

\begin{itemize}
\item
  \(u_s(t) = M\) and \(|\payt s t| < M\), therefore
  \(\dot{u}_s(t) = \beta_t \crateu t(\payt s t - u_s(t)) \leq 0\) for
  some \(\beta_t\), therefore \(u_s(t^-) \geq M\), which is absurd.
\item
  \(u_s(t) = -M\) and \(|\payt s t| < M\), therefore
  \(\dot{u}_s(t) = \beta_t \crateu t (\payt s t - u_s(t)) \geq 0\) for
  some \(\beta_t\), therefore \(u_s(t^-) \leq -M\), which is absurd.
\item
  \(\payt s t = M\), therefore:
  $$(1-\delta)r_s(x_s(t)) + \delta \sum_{s' \in S} P_{s, s'}(x_s(t))u_{s'}(t) = M$$
  But \(r_s(x_s) < M\) and \(u_{s'}(t) \leq M\) for all \(s'\),
  \newline therefore \(\sum_{s' \in S} P_{s, s'}(x_s(t))u_{s'}(t) \leq M\), \newline  so
  \(\payt s t < M\) (because \(0 < \delta < 1\)), absurd.
\end{itemize}

\end{proof}

\begin{lemma}\label{lem:diffgamma}

  Function $\pay s$ is differentiable and its differential is:
  \[\diff {\pay s} = \delta \sum_{s'} P_{s s'}(x_s)\dot u_{s'} + \ratest s t \sum_i \optgap\] In the \ref{eq:sbrd} case, $\ratest s t = 1$.\end{lemma}
  
  \begin{proof}
  
  \[\frac{d\pay s}{dt} = D_u (f_{s, \vec u(t)}(x_s(t)))(D_t u) + D_{x_s} f_{s, \vec u}(x_s)(D_t x_s)\] where $D_u$ is the partial differential in $u$.
  
  \(x_s \mapsto \stfs \cut {x_s}\) is a n-linear map in \(x_s\),
  therefore:
  $$D_{x_s} \stfs \cut {x_s}(D_t x_s) = \sum_i \stfs \cut {\dot x^i_s, x^{-i}_s}$$
  
  \(u \mapsto \stfs {u} {\cxs {} t}\) is a linear function in \(u\), and:
  \[D_u \stfs u {\cxs {} t} = \delta \sum_{s'}P_{s s'}(x_s) \dot u_s\]
  
  Therefore,
  \(\frac{d\pay s}{dt} = \delta \sum_{s'} P_{s s'}(x_s)\dot u_{s'} + \ratest s t \sum_i \optgap\).\end{proof}

\begin{lemma}\label{lem:deltalip}
  Function $\optgap$ is Lipschitz.
\end{lemma}
\begin{proof}

\(u\) is differentiable and its derivative is bounded by\newline 
\(\sup_t |\payt s t - u_s(t)|\), so \(u\) is \(2M\)-Lipschitz where
\(M\) is a bound of the \(\pay s\) and \(u_s\). The derivative of
\(x_s\) is also bounded, so it is also Lipschitz. As \(f_{s, \cdot}\) is
Lispschitz with respect to any parameter (it is multilinear). Therefore,
for all \(y\), \(t \mapsto \stfs \cut {y, \cxs {-i} t}\)
is Lispschitz with the same coefficient, so \(\optgap\) is
also Lipschitz.
\end{proof}

\hypertarget{convergence-of-the-synchronous-and-semi-synchronous-system}{%
\paragraph{Convergence of the synchronous and asynchronous
system}\label{convergence-of-the-synchronous-and-semi-synchronous-system}}

Let
\(\wstate t \in \argmin{s \in S} \left(\payt s t - u_s(t)\right)\). This means that for every $t$ we choose an arbitrary $s$ that minimizes $\payt s t - u_s(t)$. Note that as a consequence, $\payt {\wstate t} t - u_{\wstate t}(t)$ is continuous because every $\payt s t - u_s(t)$ is continuous.

\begin{lemma}[] \label{ineqgu}
There exists $A \geq 0$ such that for every $s \in \states$,
\(\payt s t - u_s(t) \geq -A \exp(\int_1^t (\delta-1)\crateu t dt)\)
\end{lemma}

\begin{proof}

By the previous lemma:

\begin{equation}\begin{aligned}\frac{d \pay s}{dt} & \geq \delta \sum_{s'}  \transSS(x_s)\crateu t \left(\payt {s'} t-u_{s'}(t)\right) \\ & \geq \delta \sum_{s'}P_{ss'}(x_s) \crateu t \left(\payt {\wstate t} t -u_{\wstate t}\right) \\ & = \delta  \crateu t \left(\payt {\wstate t} t-u_{\wstate t}\right)\end{aligned}\label{eq:bound_gama_s}\end{equation}

Moreover, for \(h > 0\):

\begin{equation}\begin{aligned}&\payt {\wstate {t+h}} {t+h} -u_{\wstate {t+h}}(t+h)- (\payt {\wstate t} t - u_{\wstate t}(t)) \\&  \geq \payt {\wstate {t+h}} {t+h} -u_{\wstate {t+h}}(t+h) - (\payt {\wstate {t+h}} t - u_{\wstate {t+h}}(t)) \\ & \geq h \min_{s\in S} \frac{d\pay s}{dt} + o(h)  + u_{\wstate{t+h}}(t)-u_{\wstate{t+h}}(t+h)\end{aligned}\label{eq:pregronwall}\end{equation}

For any $s$:
\begin{align*}u_{s}(t)-u_{s}(t+h) & = -h\frac{d u_s}{dt}+o(h) \\ & = -h\crateu t \left(\Gamma_s(t)-u_s(t)\right)+o(h)\end{align*}

Now let us suppose that $s$ is an accumulation point of $s_-(t+h)$ when $h$ goes to 0. Then, as every $\Gamma_s(t)-u_s(t)$ is continuous, we have that $\Gamma_{s_-(t)}(t)-u_{s_-(t)}(t) = \Gamma_s(t)-u_s(t)$ (else $s$ can not be an accumulation point). So, the preceding equality can be rewritten as:

$$u_{s}(t)-u_{s}(t+h) = -h\alpha(t)\left(\Gamma_{s_-(t)}(t)-u_{s_-(t)}(t)\right) + o(h)$$

This is valid for every accumulation point of $s_-(t+h)$ (and is independent of $s$) and there is a finite number of such $s$, so we also have:
\optmult{u_{s_-(t+h)}(t)-u_{s_-(t+h)}(t+h) \\= -h\alpha(t)\left(\Gamma_{s_-(t)}(t)-u_{s_-(t)}(t)\right) + o(h)}

Now, from inequality~(\ref{eq:bound_gama_s}), we have that:
$$h \min_{s\in S} \frac{d \Gamma_s}{dt} \geq h\delta \crateu t \left(\Gamma_{s_-(t)}(t)-u_{s_-(t)}\right)$$

And these two last inequalities can be summed to get:
\optmult{h \min_{s\in S} \frac{d \Gamma_s}{dt} + u_{s_-(t+h)}(t)-u_{s_-(t+h)}(t+h)+o(h) \\ \geq h(\delta-1) \crateu t \left(\Gamma_{s_-(t)}(t)-u_{s_-(t)}\right)+o(h)}

Going back to~(\ref{eq:pregronwall}):
\[\begin{split}\payt {\wstate {t+h}} {t+h} -u_{\wstate {t+h}}(t+h) \optspace{\ \ \ \ \ } \optnl - (\payt {\wstate t} t - u_{\wstate t}(t)) \\ \geq h(\delta-1) \crateu t \left(\Gamma_{s_-(t)}(t)-u_{s_-(t)}\right)+o(h)\end{split}\]

We now need a version of Grönwall Lemma that applies to this case, it is
provided here for completeness:

Let
\(v(t) = \exp\left( \int_0^t (\delta-1)\crateu t dt\right)\).

Then \(\frac{dv}{dt} = (\delta-1)\crateu t  v(t)\),
\(v(0) = 1\), \(v > 0\).

and
\(\frac{1}{v(t+h)} = \frac{1}{v(t)} - h(\delta-1)\frac{\crateu t}{ v(t)} + o(h)\)

We now proceed with the classical proof of Grönwall Lemma:

\begin{align*}
\frac{\payt {\wstate {t+h}} {t+h} - u_{\wstate {t+h}}(t+h)}{v(t+h)} & \geq \frac{\payt {\wstate t} t - u_{\wstate t}(t)}{v(t+h)} \optnl \optesp \optspace{\ \ \ \ } + h(\delta-1) \crateu t   \frac{\payt {\wstate t} t-u_{\wstate t}(t)}{v(t+h)}+o(h)\\
& \begin{multlined}\geq \frac{\payt {\wstate t} t - u_{\wstate t}(t)}{v(t)} \optnl \optesp \optspace{\ \ \ \ }  - h(\delta-1)\crateu t \frac{\payt {\wstate t} t - u_{\wstate t}(t)}{ v(t)}  \optnl \optesp \optspace{\ \ \ \ }  \\ + h(\delta-1)\crateu t\frac{\payt {\wstate t} t - u_{\wstate t}(t)}{ v(t)} +o(h)\end{multlined}\\
& \geq \frac{\payt {\wstate t} t - u_{\wstate t}(t)}{v(t)} + o(h)\end{align*}

Therefore,
\(t \mapsto \frac{\payt {\wstate t} t - u_{\wstate t}(t)}{v(t)}\) is
increasing. We can conclude:
\optmult{\payt {\wstate t} t - u_{\wstate t}(t) \\\geq \left( \payt {\wstate 0} 0 - u_{\wstate 0}(0)\right)\exp\left(\int_0^t (\delta-1)\crateu t dt\right)}
\end{proof}

\begin{lemma} The gap between $\payt s t$ and $\max_{y\in A^i}\stfs {u(t)}{y, x^{-i}_s(t)}$ converges to 0:
  \[\forall s, \sum_i \optgap \rightarrow 0\]
  \label{lem:delta0}\end{lemma}

\begin{proof}

First, we show that \(\forall i,s\),
\(\int_1^\infty \sum_{i\in I} \optgap dt < + \infty\).

Using Lemma~\ref{lem:diffgamma}:
\(\ratest s t \sum_i \optgap = \frac{d\pay s}{dt} - \delta \sum_{s'}P_{s s'}(x_s)\dot u_s\).

Therefore:
\[\forall T, \int_1^T\ratest s t \sum_i \optgap dt = \int_1^T \frac{d\pay s}{dt} - \delta\sum_{s'}\int_1^TP_{ss'}(x_s)\dot u_s\]

With the previous lemma:
\[\begin{aligned}P_{ss'}(x_s)\dot u_s & = P_{ss'}(x_s)\crateu t \left({\payt s t-u_s(t)}\right) \\ & \geq -P_{ss'}(x_s) A \crateu t \exp(\int_1^t (\delta-1)\crateu t )\end{aligned}\]

Then, for all $T$:
{ \[\begin{aligned}\beta_-\int_1^T \sum_i \optgap dt  & \leq \int_1^T \ratest s t\sum_i \optgap  dt \\ &\leq \payt s T - \payt s 1 \optnl \optesp + \delta\sum_{s'}P_{ss'}(x_s)\int_1^T A\crateu t \exp\left(\int_1^t \crateu v(\delta-1)dv\right)\\&=\payt s T-\payt s 1 \optnl \optesp + \delta\frac{A}{\delta-1} \left(\exp\left(\int_1^T \crateu v(\delta-1)dv\right)-1\right)\\&<\payt s T-\payt s 1 + \delta\frac{A}{1-\delta}\end{aligned}\]}

Then, as $\optgap$ is Lipschitz (Lemma~\ref{lem:deltalip}) and the limit of its integral is bounded and
positive, \(\optgap \limt 0\).\end{proof}

\begin{lemma}[Convergence of the synchronous and semi-asynchronous system] \label{lem:contconv}

For all \(s\in S\):
\[\payt s t = \stfs \cut {\cxs {} t} \limt \lim \sup \pay s\]
\[\text{and\ } u_s(t) \limt \lim\sup \pay s\]\end{lemma}

\begin{proof}

\[\begin{aligned} \payt s {t_2} & = \payt s {t_1} + \int_{t_1}^{t_2} \frac{d \pay s}{dt}dv \\ &\geq \payt s {t_1} + \delta\int_{t_1}^{t_2} \crateu v \left(\payt {\wstate v} v - u_{\wstate v}(v)\right) dv \\ &\geq \payt s {t_1} - A\delta \int_{t_1}^{t_2} \crateu v \exp\left(\int_1^v (\delta-1) \crateu t dt\right) dv \\ &\geq \payt s {t_1} - A\frac{\delta}{1-\delta} \exp\left(\int_1^{t_1} (\delta-1)\crateu t  dt\right)\end{aligned}\]

So \(\dfrac{A\delta}{1-\delta} \exp\left(\int_1^{t_1} (\delta-1)\crateu t dt \right)\) goes to \(0\) when
\(t_1\) goes to \(+\infty\) (thanks to hypothesis~\ref{hyp:cont}), therefore, it is sufficient to take \(t_1\)
big enough so that \(\payt s {t_1}\) is close to the \(\lim\sup\) and
the second term is small enough.

With a similar argument, \(u_s\) has a limit, and it is necessarily the
same as \(\pay s\), otherwise \(u_s\) would be unbounded (because
\(\dot u_s = (\pay s -u_s)/\crateu t\)).\end{proof}

\begin{lemma}[Convergence to the set of mixed stationary equilibria]\label{lem:contconvequilibria}

\(\{\lim \pay s\}_{s \in S}\) is an equilibrium payoff of the \(\delta\)
discounted stochastic game. \(\{x_{s}\}_{s\in \states}\) converges to
the set of mixed equilibria.\end{lemma}

\begin{proof}

Let \(\tilde x\) be an accumulation point of the vector-valued function
\(x = \{x_s\}\). Then, from Lemma~\ref{lem:delta0}:
\[\optgap = \stfs \cut {br^i_{s, u(t)}(x^{-i}_s(t)) - x^i_s(t), x^{-i}_s(t)} \rightarrow 0\]
So by continuity, for all $s$:
\[f_{s, \lim u}(br^i_{s, \lim u}(\tilde x^{-i}_s(t)) - \tilde x^i_s, \tilde x^{-i}_s) =0\]\end{proof}

\begin{proof}[Proof of Theorem \ref{thm:contasync}]\ \newline
Lemma~\ref{lem:contconvequilibria} and~\ref{lem:contconv} prove the theorems for both the \ref{eq:sbrd} and \ref{eq:sabrd} systems.
\end{proof}

%% file: include/appendix_zerosum_continuous.tex
In this section, we prove that best-response dynamics converges in two players zero-sum games, therefore extending the result of \citeauthor{leslieBestresponseDynamicsZerosum2020}. The proof is very similar to that of \cite{leslieBestresponseDynamicsZerosum2020} but it is included below for completeness. Details can be found in the other paper.

{

\input{include/notations_zerosum.tex}

Let $\{\x[s][][], \ut[][]\}_{s \in \states}$ be a solution of $\ref{eq:sabrd}$ and let $\alpha^\star > \lim_{t\rightarrow +\infty} \alpha(t)$. Note that $\alpha^\star$ may be be arbitrarily close to $0$ if the limit of $\alpha$ is 0, and indeed the case proven in \cite{leslieBestresponseDynamicsZerosum2020} is the case where $\alpha(t) = \frac{1}{t+1}$ and $\alpha^\star \rightarrow 0$.

We define the energy of the system, also known as the duality gap, as:
\begin{equation}\duality = \max_{a^1 \in A^1} \f[a^1, \x[s][2]] - \min_{a^2 \in A^2} \f[\x[s][1], a^2]
\end{equation}

Lemma~\ref{lem:bounded} states that there exists a constant $M > 0$ such that $\|\ut[s][]\|_{\infty} \leq M$ and $\|\f[]\|_{\infty} \leq M$.

By definition $\x[s][][], \ut[][]$ are differentiable almost everywhere. It is straightforward to see that it is also true for $\duality[s][]$.

We denote $\vu$ the value of the auxiliary game in state $s$ parameterized by $\ut$.

\begin{lemma}\label{lem:energy_zerosum}
For every $s \in \states$, there exists a time $T$ such that for all $t \geq T$, $|\f - \vu| \leq 4 M \alpha^\star$.
\end{lemma}
\begin{proof}
We define $D_u w_s$ and $D_{x_s} w_s$ as the partial derivatives of $\duality[s][]$ when $\ut[][]$ and $\x[s][][]$ are considered as parameter of $\duality[s][]$. With these notations, $\diff {\duality[s][]} = \dot {\ut[][]} \cdot D_u w_s + \dot {\x[s][][]} \cdot D_{x_s} w_s$.

On the one hand, using Lemma~A.2 of \cite{leslieBestresponseDynamicsZerosum2020}, $\dot {\ut[][]} \cdot D_u w_s \leq 2 \delta \max_{s'\in \states} \dot {\ut[s'][]} \leq 4\delta M \alpha(t)$. On the other hand, using~\cite{hofbauerBestResponseDynamics2006}, $\dot {\x[s][][]} \cdot D_{x_s} w_s \leq - \beta_{-} \duality$.

Therefore, $\diff{\duality[s][]} \leq -\beta_{-}\duality + 4\delta M \alpha(t)$. Since $\alpha$ is decreasing, it is arbitrarily close to $\alpha^\star$ when $t$ goes to $\infty$, so $\duality\leq 4 M \alpha^\star$ when $t$ is big enough. Note that knowing $\alpha$, a $t$ satisfying this property can be computed. This will be used in the discrete time algorithm.

Since $|\f - \vu| \leq \duality$, this gives the desired result.

\end{proof}

Define $\epsilon$ such as $\frac{(1-\delta)\epsilon}{16} = 4M\alpha^\star$ and  and $t_1(\epsilon)$ as defined in Lemma~\ref{lem:energy_zerosum}.

We define two distinguished states:
\begin{itemize}
    \item $s_f(t) \in \argmax{s\in \states}|\f - \ut[s]|$
    \item $s_v(t) \in \argmax{s\in \states}|\vu - \ut[s]|$
\end{itemize}

\newcommand{\fsf}{\f[\x[s_f(t)][]][\ut][s_f]}
\newcommand{\usf}{\ut[s_f(t)]}
\newcommand{\vsf}{\vu[s_f(t)]}

\begin{lemma}\label{lem:usf}
If $t \geq t_1(\epsilon)$, $|\usf - \fsf| \geq \epsilon$ and for an $s \in \states$, $$\left||\usf - \vsf| - |\ut[s] - \vu|\right| \leq \frac{(1-\delta)\epsilon}{8}$$
then:
$$\diff{|\ut[s] - \vu|}\leq - \frac{3(1-\delta)\crateu t \epsilon}{4}$$

\end{lemma}

\begin{proof}

First, using Lemma~A.2 of \cite{leslieBestresponseDynamicsZerosum2020}:
$$\diff {|\vu|} \leq \delta \max_{s\in \states} | \dot{\ut[][]}| = \delta \crateu t |\fsf - \usf|$$

\paragraph{If $u_s(t) \geq \vu$:}

Then $|\usf-\vsf|-\ut[s]+\vu \leq \frac{(1-\delta)\epsilon}{8}$.

$$\begin{aligned}
\diff{\ut[s]} & = \crateu t \left(\f - \ut[s]\right) \\
& \leq \crateu t\left(\f + \frac{(1-\delta)\epsilon}{8}-\vu-|\usf-\vsf|\right) \\
& \leq \crateu t \left(\frac{3(1-\delta)\epsilon}{16}-|\usf-\vsf|\right) \\
& \leq \crateu t \left(\frac{(1-\delta)\epsilon}{4}-|\usf-\fsf|\right)
\end{aligned}$$

Summing with $\vu$:
$$
\begin{aligned}
\diff{\ut[s]-\vu} & \leq\crateu t \left(\frac{(1-\delta)\epsilon}{4} + (\delta-1)|\usf-\fsf|\right) \\
& \leq \crateu t \left(\frac{(1-\delta)\epsilon}{4} - (1-\delta)\epsilon\right) \\
& \leq - \crateu t \left(\frac{3(1-\delta)\epsilon}{4}\right) \\
\end{aligned}
$$

\paragraph{If $u_s(t) \leq \vu$:} similar calculations yield the same result.

\end{proof}

We can now prove the two final lemma of this section.

\begin{lemma}
For all $s\in \states$, $\lim\sup |\ut[s]-\f| \leq 2\epsilon$.
\label{lem:zs_final}
\end{lemma}
\begin{proof}
We define $g(t) = \max\{|\usf-\vsf|, 2\epsilon\}$.

As a composition of maximum of locally Lipschitz function and as such is locally Lipschitz as well. (See Lemma~B.4 of \cite{leslieBestresponseDynamicsZerosum2020} for detailed arguments.)

Now, if $|\usf-\vsf|\leq 2 \epsilon$, then $\diff g = 0$. If $\usf-\vsf| \geq 2 \epsilon$, then, if $t$ is greater than $t^1(\epsilon)$, $|\usf-\fsf| \geq \epsilon$ (Lemma~\ref{lem:energy_zerosum}). For $t\geq t^1(\epsilon)$, on a neighbourhood of $t$, every $s$ that maximizes $|\f - \ut[s]|$ satisfies the condition of Lemma~\ref{lem:usf}, because $\f$ and $\vu$ are close thanks to Lemma~\ref{lem:energy_zerosum}. Therefore, Lemma~\ref{lem:usf} can be used and $\diff g \leq- \frac{3(1-\delta)\alpha(t)\epsilon}{4}$.

Using hypothesis~\ref{hyp:cont}, $g(t) \rightarrow 2\epsilon$, which gives the result.

\end{proof}

\begin{lemma}
For all $s \in \states$, $x_s(t)$ converge to the set of $2\epsilon$-Nash equilibria of the auxiliary game.
\end{lemma}

\begin{proof}
The previous proof gives that $\f$ is $2\epsilon$ close to $\vu$, hence the result.
\end{proof}
}

%% file: include/notations_zerosum.tex
\NewDocumentCommand{\ut}{O{} O{t}}{u_{#1}\ifthenelse{\equal{#2}{}}{}{(#2)}}
\NewDocumentCommand{\vu}{O{s} O{\ut}}{v_{#1, #2}}
\NewDocumentCommand{\x}{O{} O{i} O{t}}{x_{#1}^{#2}\ifthenelse{\equal{#3}{}}{}{(#3)}}
\NewDocumentCommand{\f}{O{\x[s][]} O{\ut} O{s}}{f_{#3, #2}\left(#1\right)}
\NewDocumentCommand{\duality}{O{s} O{t}}{w_{#1}\ifthenelse{\equal{#2}{}}{}{\left(#2\right)}}

%% file: appendix_stochastic.tex
\if\isicml0
\hypertarget{app:stochasticapprox}{
\section{Stochastic Approximations}
\label{app:stochasticapprox}}
\fi

In this section, we describe how the stochastic approximation framework with differential inclusion  \cite{benaimStochasticApproximationsDifferential2005} can be extended and used to prove result in discrete time in the autonomous case (i.e., $\alpha$ is constant).

\subsection{Correlated Asynchronous Stochastic Approximation}

An asynchronous system as defined in \citep{perkinsAsynchronousStochasticApproximation2012} is as follows. Assuming \(y_n \in \mathbb{R}^k\), one defines a system where updated
components of the vector at every step \(n\) are
\(S_n \subseteq K := [1\ldots k]\). We define
\(s^{\sharp}_{n}\) as the number of times util $n$ that $s$ occured:
$$s^{\sharp}_{n} = \sharp \{ k\ |\ s \in S_k \land 0 \leq k \leq n\}$$

We now describe now a system where component \(y_{s, n}\) is updated at rate
\(\gamma_{s^{\sharp}_n}\)if and only if \(s \in S_n\), that is:
 \begin{equation}y_{s,n+1} - y_{s, n} - \gamma_{s^{\sharp}_n} (Y_{s, n}+d_{s,n}) \in 1_{s\in S_n} \gamma_{s^{\sharp}_n} F_s(y_n)\label{eq:discreteasyncperkinsscalar}\end{equation}
where variable $Y_{s, n}$ is a random noise with $\mathbb E[Y_{s, n}]=0$ and $d_{s,n}$ goes to $0$ when $n \rightarrow \infty$.

We define:
\[\overline\gamma_n = \max_{s \in S_n}\gamma_{s^{\sharp}_n}\]
\[M_{n+1} = \text{diag}\left\{1_{s\in I_n}\frac{\gamma_{s^{\sharp}_n}}{\overline \gamma_n}\ |\ s \in K\right\}\]
and we can rewrite \eref{eq:discreteasyncperkinsscalar} to:
\begin{equation}
  y_{n+1}-y_n - \overline \gamma_n M_{n+1} (Y_n+d_n) \in \overline \gamma_n M_{n+1}F(y_n)\label{eq:discreteasyncperkins}
\end{equation}

The continuous counterpart is defined as follows. For an
\(\epsilon > 0\), $\Omega^\epsilon_{k} $ is the set of $k\times k$ diagonal matrices with coefficients between $\epsilon$ and 1:

\[\Omega^\epsilon_k := \{\text{diag}(\beta_1, \ldots, \beta_k); \beta_i \in [\epsilon, 1], \forall i = 1, \ldots, k\}\]

And the continuous system is:

\begin{equation}\diff \sax \in \overline F(\sax) := \Omega^\epsilon_{k} \cdot F(\sax)\label{eq:meanf}\end{equation}
where the multiplication is between sets (i.e., the
resulting set is the multiplication of every pair of the initial sets).

Then, the limit set of solutions of~\eref{eq:discreteasyncperkins} is internally chain transitive (see Definition~\ref{def:ict} below) for system~\eref{eq:meanf} \citep{perkinsAsynchronousStochasticApproximation2012} under assumptions stated in Subsection~\ref{sec:formalresults}.

However, we need a modified version where the asynchronoucity can be
correlated, meaning for instance that some components are updated
synchronously or that updating may be done at the same time for a set of
components. This is the case if $x^i_s$ and $u_s$ were updated at the same times for a specific state $s$ (which is an extension of this current work, as explained in the conclusion) or for $u_s$ which is always updated at every step in~\ref{eq:safp}. Therefore, we now suppose that every
\(S_n \in \mathcal S \subseteq K\). For instance, if the \(s\)
component is updated at every step, it can be expressed with
\(\forall S' \in \mathcal S\), \(s \in S'\). Then we define an
alternative set of diagonal matrices for the continuous version:
\(\Omega^\epsilon_{k, \mathcal S} := \text{diag}(\text{conv}(\mathcal S) \bigcap [\epsilon, 1]^K)\)
and the map $\overline F(\sax) := \Omega^\epsilon_{k, \mathcal S} \cdot F(\sax)\label{eq:meanfc}$

Then we can link the internally chain transitive sets of
differential inclusion $\diff \sax \in \overline F(\sax)$ and limit sets of solutions of~(\ref{eq:discreteasyncperkins}). As systems~\ref{eq:sabrd} and~\ref{eq:sbrd} can be written as $\overline F$ with a suitable $\mathcal S$ and $F$, making it possible to prove the rest of Theorem~\ref{thm:discrasync} using the convergence results of the continuous time systems of the previous section, see section~\ref{sec:convfpasync}.

\hypertarget{sec:formalresults}{
\subsection{Formal Results}
\label{sec:formalresults}}

We start with the definition of Marchaud maps. They are used in most stochastic approximation theorems, even if the term is not always employed. In our systems, as the best-response map $\text{br}$ is piecewise constant and the rest of the right hand side is continuous, right hand sides of the differential inclusions are Marchaud maps.

\begin{definition}[Marchaud map]\label{def:marchaud}
\(F: \mathbb R^{K} \rightrightarrows \mathbb R^K\) is a Marchaud map if:
\begin{enumerate}
\def\labelenumi{(\roman{enumi})}
\item
  F is a closed set-valued map, \emph{i.e.}
  \(\{(x, y) \in \mathbb R ^ {K} \times \mathbb R ^ K\ |\ y \in F(x)\}\)
  is closed.
\item
  for all \(\sax \in \mathbb R ^ {K}\), $F(\sax)$ is a non-empty, compact, convex
  subset of \(\mathbb R ^K\)
\item
  there exists \(c>0\) such that
  \(\sup_{\sax\in \mathbb R ^ {K} z\in F(\sax)} ||z|| \leq c(1+||\sax||)\)
\end{enumerate}
\end{definition}

We now need the definition of internally chain transitive sets, as stated in \cite{benaimStochasticApproximationsDifferential2005}. They will later be used to characterize the limit sets of the discrete time systems.

\begin{definition}[Internally chain transitive]
A set $A$ is internally chain transitive for a differential inclusion $\diff \sax \in F(\sax)$ if it is compact and if for all $\sax, \sax' \in A$, $\epsilon > 0$ and $T > 0$ there exists an integer $n\in \mathbb N$, solutions $\sax_1, \ldots \sax_n$ to the differential inclusion and real numbers $t_1, t_2, \ldots, t_n$ greater than $T$ such that:
\begin{itemize}
    \item $\sax_i(s) \in A$ for $0\leq s \leq t_i$
    \item $\|\sax_i(t_i)-\sax_{i+1}(0)\|\leq \epsilon$
    \item $\|\sax_1(0)-\sax\| \leq \epsilon$ and $\|\sax_n(t_n)-\sax'\|\leq \epsilon$
\end{itemize}
\label{def:ict}
\end{definition}

\begin{definition}[Asymptotic pseudo-trajectories]
  A continuous function $z:\mathbb R^+\rightarrow \mathbb R^m$ is an asymptotic pseudo-trajectory of a differential inclusion if\ $\lim_{t\rightarrow +\infty}\mathbf D(\Theta^t(z), S) = 0$ where $\Theta^t(z)(s) = z(t+s)$ (it is the translation operator), $S$ is the set of all solutions of the differential inclusion and $\mathbf D$ is the distance between continuous functions defined as:

  $$\mathbf D(f, g) :=\sum_{k=1}^\infty \frac{1}{2^k} \min(\|f-g\|_{[-k, k]}, 1)$$
  where $\|\cdot\|_{[-k,k]}$ is the supremum norm on the interval $[-k, k]$.
  \label{def:apt}
\end{definition}

This two last definitions will be useful with Theorem~4.3 of \cite{benaimStochasticApproximationsDifferential2005} that establishes that the limit set of asymptotic pseudo-trajectories is internally chain transitive. What is left to prove is that an affine interpolation of the discrete time system is an asymptotic pseudo-trajectories. Below is the proof for the synchronous system \ref{eq:sfp} and the next section deals with semi-asynchronous and fully-asynchro\-nous systems.

\begin{lemma}\label{lem:sfpict} The limit set of \ref{eq:sfp} is internally chain transitive with respect to \ref{eq:sbrd} for $\crateu t = 1$ for all $t$.
\end{lemma}
\begin{proof}
Proposition 1.3 and Theorem 4.2 of \cite{benaimStochasticApproximationsDifferential2005} establish that the affine interpolation of sequences $x_{n, s}, u_n$ is a perturbed solution and then an asymptotic pseudo trajectory. Theorem 4.3 from the same article proves that the limit set is internally chain transitive.
\end{proof}

\subsection{Correlated Asynchronous Stochastic Approximations}

We now extend a theorem originally proven by \citeauthor{perkinsAsynchronousStochasticApproximation2012}:

\begin{theorem}[Analog of Theorem 3.1 of~\cite{perkinsAsynchronousStochasticApproximation2012}]\label{thm:asyncapprox}
Suppose that:
\begin{enumerate}[label=(\roman*)]
    \item $\sax_n \in C$ for all $n$ where $C$ is compact
    \item The set valued application $F:C \rightrightarrows C$ is Marchaud
    \item Sequence $\gamma_n$ is such that \begin{enumerate}
        \item $\sum_n \gamma_n = \infty$ and $\gamma_n \xrightarrow[n\rightarrow \infty]{} 0$
        \item for $x \in(0, 1), \sup_n \gamma_{\left[xn\right]}/\gamma_n < A_x < \infty$ where $\left[\cdot \right]$ is the floor function.
        \item for all $n$, $\gamma_n \geq \gamma_{n+1}$
    \end{enumerate}
    \item \begin{enumerate}
        \item For all $\sax\in C$, $\mathcal S_n, \mathcal S_{n+1} \in \mathcal S$,
        \optmult{\mathbb P(S_{n+1} = \mathcal S_{n+1}|\mathcal F_n)= \\ \mathbb P(S_{n+1}=\mathcal S_{n+1}|S_n = \mathcal S_n, \sax_n=\sax)}
        \item The probability transition between $\mathcal S_n$ and $\mathcal S_{n+1}$ is Lipsichitz continuous in $x_n$ and the Markov chain that $S_n$ form is aperiodic, irreducible and for every $s \in \mathcal S$, there exists $S \in \mathcal S$ such that $s \in S$.
    \end{enumerate}
    \item For all $n$, $Y_{n+1}$ and $S_{n+1}$ are uncorrelated given $\mathcal F_n$
    \item For some $q \geq 2$, $\left\{\begin{aligned}&\sum_n \gamma_n^{1+q/2} < \infty \\ & \sup_n \mathbb E(\|Y_n\|^q)<\infty\end{aligned}\right.$
    \item $d_n\rightarrow 0$ when $n\rightarrow \infty$
    
\end{enumerate}
Then with probability 1, affine interpolation $\overline \sax$ is an asymptotic pseudo-trajectory to the differential inclusion,
$$\diff \sax \in \overline F(\sax)$$
where $\left\{\begin{aligned}&\overline F(\sax) := \Omega^\epsilon_{k, \sigma} \cdot F(\sax)\\& \Omega^\epsilon_k :=  \left\{\text{diag}(\beta_1, \ldots, \beta_k) | \forall i \in \left\{1, \ldots, k\right\} \land \beta_i \in [\epsilon, 1]  \right\} \\& \epsilon > 0\end{aligned}\right.$
\end{theorem}

However, at every step of the proof of \citeauthor*{perkinsAsynchronousStochasticApproximation2012}, we can take into account that $S_n \in \mathcal S$, therefore we do not need every matrix $\text{diag}([\epsilon, 1]^K)$ in $\Omega^\epsilon_k$ but only those that are also in $\text{diag}(\text{conv}_\epsilon(\mathcal S))$ where $\text{conv}(\mathcal S)$ is the convex hull of $\mathcal S$ composed with $\max(\epsilon, \cdot)$ for every coordinate. Indeed, when update rates are manipulated, they are summed via integrals or floored by an $\epsilon > 0$. The resulting vectors belongs to $\text{conv}_\epsilon(\mathcal S)$ at every step of the proof. Therefore, the conclusion of the theorem can be changed with \(\Omega^\epsilon_{k, \mathcal S} := \text{diag}(\text{conv}(\mathcal S) \bigcap [\epsilon, 1]^K)\). This makes it possible to use the Theorem in our asynchronous and semi-asynchronous cases.

The full proof is included in Section \ref{sec:fullproof}.

\hypertarget{sec:convfpasync}{
\subsection{Convergence of a Fictitious Play procedure in identical interest stochastic games}
\label{sec:convfpasync}}

In this subsection, we first characterize the internally chain transitive sets of $\ref{eq:sbrd}$ and $\ref{eq:sabrd}$ before using this characterization to prove a second time the convergence of FP in identical interest stochastic games.

\begin{lemma}[Internally Chain Transitive Sets]\label{lem:internallychaintransitive}
If for all $t$, $\crateu t = 1$ and if $L$ is internally chain transitive either for~\ref{eq:sabrd} then $$L \subseteq \left\{(x, u) | \begin{aligned} \forall s \in \states\ \forall i \in I,\ \stfs u {x_s} = u_s \\ \land\ x^i_s \in \argmax{y^i \in A^i} \stfs u {y^i, x^{-i}_s}\end{aligned} \right\}$$
\end{lemma}

\begin{proof}\ 

We define:
\begin{align*}
  A := & \left\{(x, u) \left|
          \begin{aligned}\forall s \in \states\ \forall i \in I, \ \stfs u {x_s} = u_s  \\ \land\  x^i_s \in \argmax{y^i \in A^i} \stfs u {y^i, x^{-i}_s}\\\end{aligned}\right.
        \right\} \\
  B := & \left\{(x, u) \left|
          \forall s \in \states\ \stfs u {x_s} \geq u_s\right.
        \right\}
\end{align*}

We first show that $L \subseteq B$. In order to do that, we take an element of $L$ and show that any path starting from this element is brought towards $B$, leading to the fact that the element is necessarily already in $B$ (by definition of internal chain transitivity).

\newcommand{\minf}{\zeta}

Let $(x, u) \in L$ and suppose that $(x, u) \not\in B$, that is: \[-\minf := \min_{s\in \states}\stfs u {x_s} - u_s < 0\]

Then for the case of \ref{eq:sbrd}, for any $T > 0$, there exists $n \in \mathbb N$, solutions of \ref{eq:sbrd} $(x_1, u_1), \ldots (x_n, u_n)$ and $t_1, \ldots, t_n$ greater than $T$ as in Definition~\ref{def:ict} for $\epsilon = \minf/2$.

Then $\min_{s\in \states}\stfs {u_1(0)} {x_{1,s}(0)} - u_{1,s}(0) \geq -\minf -\minf/2$.

Now we can use Lemma~\ref{ineqgu} with $\crateu t = 1$, for all $s$:
\optmult{\stfs {u_1(t_1)} {x_{1,s}(t_1)} - u_{1,s}(t_1) \geq \\ (\stfs {u_1(t_1)} {x_{1,s}(t_1)} - u_{1,s}(t_1)) \exp((\delta-1)t_1) \geq \\ (-\frac{3}{2}\minf) \exp((\delta-1)T)}
So for $T$ big enough, then for all $s$: 
\[\stfs {u_1(t_1)} {x_{1,s}(t_1)} - u_{1,s}(t_1) \geq -\minf/4\]
Iteratively, we get:
 \[\stfs {u_n(t_n)} {x_{n, s}(t_n)}-u_{n, s}(t_n) \geq -\minf/4\] which is contradictory to the fact that $\min_{s\in \states}\stfs u {x_s} - u_s = -\minf$.

For \ref{eq:sabrd}, we have the exact same proof.

So $L \subseteq B$.

We can now use a more classic argument to show that $L \subseteq A$ with a Lyapunov function now that the ambient space can be restricted to $B$. Let us define $V(x, u):= \sum\limits_{s\in \states} \stfs u {x_s}$. Then, $V$ is a Lyapunov function for set $A$ with ambient space $B$. Indeed, on $B$, $\diff {u_s} \geq 0$, so $\diff {\stfs u {x_s}} \geq 0$ (with Lemma~\ref{lem:diffgamma}). Therefore $\diff {\stfs u {x_s}} = 0$ for every $s$ if and only if $(x, u) \in A$. Moreover, $V(A)$ has empty interior thanks to Sard's Theorem.

So we can use Proposition 3.27 of \cite{benaimStochasticApproximationsDifferential2005}: it applies in case the Lyapunov function is defined on invariant set. So $L$ is contained in~$A$.

\end{proof}

\paragraph{A second proof of Theorem~\ref{thm:discrasync} using continuous time}

Below, we prove a second time Theorem~\ref{thm:discrasync} in the autonomous case (that is, $\durate = 1$) to show that our extension to the stochastic approximations framework is directly useful.

\begin{proof}[Proof of Theorem~\ref{thm:discrasync}]

For systems~(\ref{eq:safp}) we now need to apply Theorem~\ref{thm:asyncapprox}. Variable $Y_n$ is $0$ in our case because there is no noise. $S_{n+1}$ is the next state variable and it has distribution $P_{S_n}(a_n)$. We check the assumptions:
\begin{itemize}
    \item (i) is guaranteed because every variable of the system is bounded.
    \item (ii) is guaranteed because the best-response map is marchaud and the derivative of $u$ is continuous.
    \item for (iii) and (vi) we use $\gamma(n) = 1/n$, so every assumption is trivial to verify.
    \item (iv) and (v) comes from the definition of a play and the ergodicity hypothesis on the game
\end{itemize}

Therefore the affine interpolation of a sequence of fictitious play for stochastic games under our assumption is an asymptotic pseudo-trajectory, which implies that its limit set is internally chain transitive by Theorem 4.3 of \cite{benaimStochasticApproximationsDifferential2005}.

For system~\ref{eq:sfp} Lemma~\ref{lem:sfpict} states that the limit set is internally chain transitive.

Then Lemma~\ref{lem:internallychaintransitive} concludes the proof: the limit set is internally chain transitive and consequently included in the set of equilibria.

\end{proof}

\subsection{Convergence of FP in zero-sum stochastic games}

We consider \ref{eq:sbrd} and \ref{eq:sabrd} in zero-sum stochastic games in the autonomous case, that is for $\alpha(t) = \alpha^\star$

\begin{proof}[Proof of Theorem~\ref{thm:discrzs}]

The previous proof checks all the hypothesis necessary to apply the stochastic approximation framework we presented in this section. Therefore, a characterization of internally chain transitive sets will be sufficient to conclude.

In the proof of Lemma~\ref{lem:energy_zerosum}, we showed that function $\max{w_s(t)-\frac{1}{\beta_-}4\delta M \alpha^\star}$ is a Lyapunov function. Therefore, the set where the duality gap is lower or equal than $4\delta M\alpha^\star$ is a internally chain transitive set. Furthermore, in the proof of Lemma~\ref{lem:zs_final}, we defined a function $g$ which is a Lyapunov function relative to the previous internally chain transitive set.

\end{proof}

\subsection{Different Priors and Team Games}\label{sec:priors_team_games}

In this subsection, we suppose that every player has its own $u^i$ estimates. We are going to show that in this case, internally chain transitive sets are included into the set where the estimates $u^i$ are equal (up to a constant) for every $i$.

Indeed, suppose that $G$ is now a team game. Then every $r^i_s$ can be written $r^i_s = r_s + M_i$ with the convention that $M_1 = 0$ and $r_s = r^1_s$. Then let us show that any internally chain transitive set $L$ is included in $\{(x, u)\ |\ u^i_s = u^1_s + M_i \ \forall i, s \}$.

Define $V^i(x, u) = \arg\max_s |u^i_s - u^1_s - M_i|$

Let $s$ that maximizes $|u^i_s - u^1_s - M_i|$ so that $V^i(x, u) = |u^i_s - u^1_s - M_i|$.

Then if $u^i_s> u^1_s - M_i$, then $V^i(x(t), u(t))$ can be differentiated for almost every $t$ (using the same techniques as in Section~\ref{app:continuous_proof}):

\begin{equation*}
\begin{aligned}
\diff {V^i} & = \alpha(t) (f^i_{s, u^i}(x_s(t)) - f^1_{s, u^i}(x_s(t)) - u^i_s(t) + u^1_s(t) \\
& \leq \alpha(t) ((1-\delta) M_i + \delta V^i(x, u) + \delta M_i - u^i_s(t) + u^1_s(t)
& \leq \alpha(t)(\delta-1)V^i(x, u)
\end{aligned}
\end{equation*}

And similar calculations for the case $u^i_s \leq u^1_s - M_i$ give the same results.

Therefore $V^i$ is a Lyapunov function and $L \subseteq {V^i}^{-1}(\{0\})$, hence the result.

\hypertarget{sec:fullproof}{
  \section{Proof of Theorem~\ref{thm:asyncapprox}}
  \label{sec:fullproof}
}

In this subsection, we show a proof of Theorem~\ref{thm:asyncapprox}. It is a modification of Theorem 3.1 of \cite{perkinsAsynchronousStochasticApproximation2012}. In order to carry the proof, we first need a general theorem found in \cite{benaimStochasticApproximationsDifferential2005}:

\begin{theorem}[Linear interpolation are asyptotic pseudo-trajectories]\label{thm:linapt}
  Consider the stochastic approximation process
\begin{equation}
  \sax_{n+1}-\sax_n \in \gamma_n\left[F(\sax_n)+Y_{n+1}+d_{n+1}\right]
  \label{eq:full:discrete}
\end{equation}

under the assumptions:
\begin{enumerate}[label=(\roman*)]
  \item For all $T> 0$
  \begin{equation}
    \lim_{n\rightarrow \infty}\sup_{k}\left\{\left\|\sum_{i=n}^{k-1}\gamma_{i+1}Y_{i+1}\right\|; k=n+1,\ldots,m(\tau_n+T)\right\}=0
    \label{eq:full:limsup}
  \end{equation}

  where $\tau_0=0$, $\tau_n=\sum_{i=1}^n\\gamma_i$ and $m(t)=\sup\{k\geq 0; t \geq \tau_k\}$,
  \item $\tau_n \xrightarrow[n\rightarrow \infty]{} \infty$ and $\gamma_n \xrightarrow[n\rightarrow \infty]{} 0$
  \item $\sup_n\|\sax_n\|=\mathcal Y < \infty$
  \item $F$ is a Marchaud map
  \item $d_n \rightarrow 0$ as $n\rightarrow \infty$ and $\sup_n\|d_n\|=d<\infty$
\end{enumerate}
Then a linear interpolation of the iterative process $\{\sax_n\}_{n\in\mathcal N}$ given by (\ref{eq:full:discrete}) is an asymptotic pseudo-trajectory of the differential inclusion
\begin{equation}
  \diff x \in F(x)
  \label{eq:full:di}
\end{equation}
\end{theorem}

\begin{proof}[Proof of Theorem~\ref{thm:asyncapprox}]

  We are going to use Theorem~\ref{thm:linapt} and the four conditions must be verified for stochastic process~\ref{eq:discreteasyncperkins} so as its linear interpolation is an asymptotic pseudo-trajectory of \ref{eq:meanf}.

  To do this, we first define the discrete time system that Theorem~\ref{thm:linapt} will be applied to. We define $\tilde M_n := \operatorname{diag}(\max\{1_{s\in I_n}\frac{\gamma_{s^\sharp_n}}{\overline \gamma_n}, \epsilon\})$. Note that, consequently, $\tilde M_n \in \text{diag}(\text{conv}(\mathcal S) \bigcap [\epsilon, 1]^K) = \Omega^\epsilon_{k, \mathcal S}$. 
   We select $f_n \in F(x_n)$ in the differential inclusion so as for every $n$, $\sax_{n+1}=\sax_n + \overline \gamma_{n+1}M_{n+1}\left[f_n+Y_{n+1}+d_{n+1}\right]$. 
   Then define $\overline Y_{n+1} := f_n(M_{n+1}-\tilde M_{n+1})+M_{n+1}V_{n+1}$, that is to say that $\overline Y_{n+1}$ 
   is the noise $Y_{n+1}$ plus the error induced by the fact that every state is updated 
   at a minimum $\epsilon$ rate. Then we have 
   $\sax_{n+1}=\sax_n + \overline \gamma_{n+1}\left[\tilde M_{n+1} f_n + \overline Y_{n+1}+\overline d_{n+1}\right]$\LB{$\overline d_n$}.

  So $\sax_{n+1}-\sax_n \in \overline \gamma_{n+1}\left(\Omega^\epsilon_{K, \mathcal S}\cdot F(\sax_n) + \overline Y_{n+1} + \overline d_{n+1}\right)$.

  And now we verify assumptions of Theorem~\ref{thm:linapt}:

  \begin{enumerate}[label=(\roman*)]
    \item For $T>0$:
    \begin{equation*}
    \begin{multlined}
      \sup_{k}\left\{\left\|\sum_{i=n}^{k-1}\overline \gamma_{i+1} \overline Y_{i+1}\right\|; k = n+1,\ldots,\overline m(\overline \tau_n+T)\right\} \\ \leq \sup_{k}\left\{\left\|\sum_{i=n}^{k-1}\overline \gamma_{i+1} M_{i+1} Y_{i+1}\right\|; k = n+1,\ldots,\overline m(\overline \tau_n+T) \right\} \\ + \sup_{k}\left\{ \left\|\sum_{i=n}^{k-1}\overline \gamma_{i+1} f_i (M_{i+1} - \tilde M_{i+1})\right\|; k = n+1,\ldots,\overline m(\overline \tau_n+T)\right\}
    \end{multlined}
    \end{equation*}

    The first part of the sum goes to $0$ via classical Kushner-Clark condition and assumptions (iii) and (vi), the proof is detailled in Lemma 3.3 of \cite{perkinsAsynchronousStochasticApproximation2012}. Regarding the second part, it is exactly Lemma 3.6 of \cite{perkinsAsynchronousStochasticApproximation2012} and this applies because of assumptions (iii), (iv) and (v).
    \item This is assumption (iii).
    \item This is assumption (i) of Theorem~\ref{thm:linapt}.
    \item The map is  $\overline F(\sax) := \Omega^\epsilon_{k, \mathcal S} \cdot F(\sax)$ and it is Marchaud because $F$ is Marchaud (assumption (ii)) and $\Omega^\epsilon_{k, \mathcal S}$ is compact (so every property of Definition~\ref{def:marchaud} holds).
    \item This is assumption (vii).
  \end{enumerate}

  So Theorem~\ref{thm:linapt} applies and gives the desired result.

\end{proof}

%% file: include/appendix_technical_lemmas.tex
Discrete-time Grönwall can be found in the literature with various assumptions. For the sake of completeness, we include here a version that matches the assumptions we have in our paper, with the associated proof. It is a differential version with error terms.

{
\NewDocumentCommand{\g}{O{n}}{g_{#1}}
\NewDocumentCommand{\y}{O{n}}{y_{#1}}
\NewDocumentCommand{\bn}{O{n}}{b_{#1}}
\NewDocumentCommand{\vn}{O{n}}{v_{#1}}
\NewDocumentCommand{\prodg}{O{n}}{\Pi_{k=0}^{#1}(1+\g[k])}
\NewDocumentCommand{\sumb}{O{n}}{\sum_{k=0}^{#1} \bn[k]}

\begin{lemma}[Discrete-Time Grönwall]
Let $\{\y\}$, $\{\g\}$, $\{\bn\}$ sequences of real numbers such that $1 > 1+\g > 0$ for all $n$ and:
$$\y[n+1]-\y \leq \g[n+1] \y + \bn[n+1]$$

Then $\y \leq y_0 \prodg + \sum_{k=0}^n \bn[k]$.
\label{lem:discretegronwall}
\end{lemma}
\begin{proof}
    We define $\vn := \frac{\y - \sum_{k=0}^n \bn[k] }{\prodg}$. We show that $\vn$ is decreasing:

    $$\begin{aligned}
        \vn[n+1] - \vn  & = \frac{\y[n+1]-\y}{\prodg[n+1]} + \frac{\y}{\prodg[n+1]} - \frac{\y}{\prodg} - \frac{\sumb[n+1]}{\prodg[n+1]} + \frac{\sumb}{\prodg} \\
        & \leq \frac{\g[n+1] \y + \bn[n+1]}{\prodg} + \frac{\y}{\prodg}\left(\frac{1}{1+ \g[n+1]} -1\right) - \frac{\sumb[n+1]}{\prodg[n+1]} + \frac{\sumb}{\prodg}  \\
        & \leq \frac{\g[n+1] \y + \bn[n+1]}{\prodg} + \frac{\y}{\prodg} \frac{-\g[n+1]}{1+\g[n+1]}  - \frac{\sumb[n+1]}{\prodg[n+1]} + \frac{\sumb}{\prodg}  \\
        & \leq \frac{\bn[n+1]}{\prodg[n+1]} - \frac{\sumb[n+1]}{\prodg[n+1]} + \frac{\sumb}{\prodg} \\
        & \leq \frac{\sumb}{\prodg}\left(1 - \frac{1}{1+\g[n+1]}\right) \\
        & \leq 0
    \end{aligned}$$

    And $\vn[0] = y_0$, hence the result.

\end{proof}
}

\begin{lemma}[Bound on $\durate$]\label{lem:bounddiscralpha} Under hypothesis \ref{hyp:discr}, for all $n$, $\frac{\durate}{\sdurate} \leq \frac{1}{n+1}$.
\end{lemma}
\begin{proof}
By induction: for $n=0$, $\frac{\durate}{\sdurate} = 1$. Now for $n+1$: \begin{equation*}\frac{\durate[n+1]}{\sdurate[n+1]} \leq \frac{\durate[n]}{\sdurate[n]}\frac{\sdurate[n]}{\sdurate[n+1]} \leq \frac{1}{n+1}\left(1 - \frac{\durate[n+1]}{\sdurate[n+1]}\right)\end{equation*}

As a consequence:
\begin{equation*}
    \frac{\durate[n+1]}{\sdurate[n+1]}\frac{n+2}{n+1} \leq \frac{1}{n+1}
\end{equation*}
And the result follows.
\end{proof}